\documentclass[11pt,fleqn,letterpaper]{article}
\usepackage{fullpage,amsmath,amssymb,amsthm}
\usepackage{xspace}
\usepackage{multirow}
\usepackage[ruled,vlined,linesnumbered]{algorithm2e}
\usepackage[hypertex]{hyperref}

\def\tO{\tilde{O}}
\def\tOmega{\tilde{\Omega}}
\def\ala{\`a la\xspace}
\newcommand{\REMOVE}[1]{}

\newcommand{\veps}{\varepsilon}

\newcommand{\minn}[1]{\min\{#1\}}
\newcommand{\maxx}[1]{\max\{#1\}}
\newcommand{\aset}[1]{\{#1\}}

\def\Schaffer{Sch{\"a}ffer}

\DeclareMathOperator{\EX}{\mathbb E}

\newtheorem{theorem}{Theorem}[section]
\newtheorem{lemma}[theorem]{Lemma}

\newtheorem{claim}[theorem]{Claim}

\theoremstyle{definition}

\def\Lovasz{Lov\'asz\xspace}

\newcounter{this-list}

\newcounter{par-list}
\newlength{\parlistlength}
\begin{document}

%\pagestyle{myheadings}
%\addtolength{\headsep}{0.3in}

\title{Directed Spanners via Flow-Based Linear Programs}
\author{Michael Dinitz\thanks{Email: \href{mailto:michael.dinitz@weizmann.ac.il}{\texttt{michael.dinitz@weizmann.ac.il}}} \qquad\qquad
Robert Krauthgamer\thanks{Supported in part by The Israel Science Foundation (grant \#452/08), and by a Minerva grant.
Email: \href{mailto:robert.krauthgamer@weizmann.ac.il}{\texttt{robert.krauthgamer@weizmann.ac.il}}
}
\\
Weizmann Institute of Science
}

\begin{titlepage}
\maketitle
\thispagestyle{empty}

\begin{abstract}
We examine directed spanners through flow-based linear programming relaxations.
We design an $\tilde{O}(n^{2/3})$-approximation algorithm
for the directed $k$-spanner problem that works for all $k\geq 1$,
which is the first sublinear approximation for arbitrary edge-lengths.
Even in the more restricted setting of unit edge-lengths,
our algorithm improves over the previous $\tO(n^{1-1/k})$ approximation \cite{BGJRW09} when $k\ge 4$.
For the special case of $k=3$
we design a different algorithm achieving an $\tilde{O}(\sqrt{n})$-approximation,
improving the previous $\tilde{O}(n^{2/3})$ \cite{EP05,BGJRW09}.
Both of our algorithms easily extend to the \emph{fault-tolerant} setting,
which has recently attracted attention but not from an approximation viewpoint.
We also prove a nearly matching integrality gap of
$\tOmega(n^{\frac13 - \epsilon})$ for any constant $\epsilon > 0$.

A virtue of all our algorithms is that they are relatively simple.
Technically, we introduce a new yet natural flow-based relaxation,
and show how to approximately solve it even when its size is not polynomial.
The main challenge is to design a rounding scheme
that ``coordinates'' the choices of flow-paths between the many demand pairs
while using few edges overall.
We achieve this, roughly speaking,
by randomization at the level of \emph{vertices}.
\end{abstract}
\end{titlepage}

\section{Introduction}

We examine several directed spanner problems from the perspective
of approximation via a linear programming (LP) relaxation.
In particular, we design for these classical NP-hard problems
flow-based LP relaxations,
and then investigate how well these relaxations approximate the optimal spanner,
providing nearly matching upper and lower bounds.
We begin by introducing the spanner problems that we consider,
focusing throughout on \emph{directed} graphs;
we briefly compare to undirected graphs in Section \ref{sec:conclusion}.

\subsection{Spanner Problems}

Let $G=(V,E)$ be a a strongly connected directed graph.%
\footnote{The assumption of strong connectivity is for notational convenience,
although the definitions and all our results extend easily to all digraphs.}
A \emph{$k$-spanner} of $G$, for $k\ge 1$,
is a subgraph $G'=(V,E')$,
that preserves all pairwise distances within factor $k$,
i.e. for all $u,v\in V$,
\begin{equation} \label{eq:defn}
  d_{G'}(u,v) \leq k\cdot d_G(u,v).
\end{equation}
Here and throughout, $d_H$ denotes the shortest-path distance in a graph $H$.
It is easy to see that requiring \eqref{eq:defn} only for edges $(u,v)\in E$
suffices.

In the directed $k$-spanner problem with unit lengths,
the input is the graph $G$,
and the goal is to find a $k$-spanner $G'$ having the minimum number of edges.
We allow the \emph{stretch} $k$ to be a function of $n=|V|$,
e.g. $k=O(\log n)$, and in fact some of our results are most interesting when $k = \Omega(\log n)$.
This definition was introduced by Peleg and \Schaffer~\cite{PS89}
(in particular, they showed the problem is NP-hard),
and since then it has been studied extensively,
with applications ranging from routing in networks (e.g. \cite{AP95,TZ05})
to solving linear systems (e.g. \cite{ST04a,EEST08}).

The above definition has several natural generalizations.
An obvious one is to let $G$ have nonnegative edge-lengths,
leading to more complicated distances.
This is the directed $k$-spanner problem \emph{with arbitrary edge-lengths}.
Another generalization, introduced in \cite{CLPR09},
incorporates fault-tolerance:
 a $k$-spanner $G'$ is \emph{$r$-vertex-tolerant} if for all $F \subseteq V$ with $|F| \leq r$ we have that $G' \setminus F$ is a $k$-spanner of $G \setminus F$.  The definition of \emph{$r$-edge-tolerant} is the same, except that $F$ is a subset of $E$ rather than of $V$.
Clearly, the special case $r=0$ is just the standard notion defined above.
This paper address both of these generalizations.

Yet another generalization of the problem is the client-server
model~\cite{EP01}: the input contains also a set $\mathcal C
\subseteq E$ of so-called client edges and a set $\mathcal S
\subseteq E$ of server edges,
the requirement \eqref{eq:defn} is only needed for edges in $\mathcal C$,
while the spanner is only allowed to use
edges in $\mathcal S$ (i.e.~$E' \subseteq \mathcal S$).
Obviously, the case $\mathcal C=\mathcal S=E$ is just the standard notion
defined above.
Our results extend to this model in a straightforward manner, but for
the sake of exposition, we shall not address it directly.

\subsection{Results}

We first present a flow-based LP relaxation for spanner problems
(Section \ref{sec:LP}).
This relaxation is quite natural but appears to be new, and in particular it differs from the ones used in \cite{DK99,BGJRW09}.  We then use this LP relaxation to obtain the approximation algorithms
described below (see also Table~\ref{tab:results}).

\paragraph{General stretch $\mathbf{k}$.}

Our first algorithmic result is an $\tilde{O}(n^{2/3})$-approximation
for the directed $k$-spanner problem that works for all $k \geq 1$,
even with arbitrary edge-lengths (Section~\ref{sec:alg_main}).  This is the first approximation algorithm that handles the more general case of arbitrary edge-lengths. And even for unit edge-lengths, it improves over the previously known
$\tilde{O}(n^{1-1/k})$-approximation for general $k\ge 3$ due to
Bhattacharyya, Grigorescu, Jung, Raskhodnikova, and Woodruff~\cite{BGJRW09}.
Thus our result shows that the approximation need not increase with $k$,
and provides the first sublinear (in $n$) approximation ratio for
$k \geq \log n$.  The cases $k=2,3$ are addressed separately below.
Furthermore, using the reduction of~\cite{BGJRW09}
from Transitive-Closure $k$-spanner to directed $k$-spanner
we obtain for the former problem an $\tilde{O}(n^{2/3})$-approximation,
improving over their $\tO(\min\{n^{1-1/k},n/k^2\})$-approximation
for all $k\ll n^{1/6}$.

We complement the above algorithmic result by
showing that our (rather natural) LP relaxation has an integrality gap
of $\tilde{\Omega}(\frac1k n^{1/3 - \epsilon})$ for every constant $\epsilon > 0$,
even in the unit-length case (Section~\ref{sec:integrality_gap}).
Previously, Elkin and Peleg~\cite{EP07} proved that for
every fixed $0 < \epsilon, \delta < 1$ and $3 \leq k = o(n^{\delta})$,
approximating the directed $k$-spanner with unit edge-lengths problem
within ratio $2^{\log^{1-\epsilon} n}$ is quasi-NP-hard
(similar hardness results were already known for smaller ranges
of $k$~\cite{Kortsarz01}).
We conclude that a polynomial approximation (independent of $k$)
is probably the best one can hope for, and specifically the best possible
exponent appears to be in the range $[1/3,2/3]$.

\paragraph{Stretch $\mathbf{k=3}$.}

For directed $3$-spanner with unit-length edges
we achieve an even better $\tO(\sqrt n)$-approximation
(Section \ref{sec:stretch3}).
Notice that this approximation factor matches, up to lower order factors,
the $O(\sqrt n)$-approximation known for undirected graphs
(an immediate consequence of the absolute guarantee of \cite{ADDJS93}
that every undirected graph has a $k$-spanner with $O(n^{1+2/(k+1)})$ edges).  The previous approximation known for this case is $\tO(n^{2/3})$,
first proved by Elkin and Peleg~\cite{EP05}.
A similar approximation can be obtained also by the aforementioned
algorithm of \cite{BGJRW09}, and by our first algorithm mentioned above.

\paragraph{Stretch $\mathbf{k=2}$.}

This case (directed $2$-spanner with unit-length edges)
is rather exceptional and is known to have tight approximation bounds:
$O(\log n)$ approximation \cite{KP94,EP01}
and $\Omega(\log n)$ NP-hardness \cite{Kortsarz01}.%
\footnote{It is possible to refine the approximation
in terms of the graph's average/maximum degree.}
We show similar bounds on the integrality gap of our LP relaxation
(in Section~\ref{sec:2-spanner}),
a finding that is not very surprising but affirms the strong connection
between our LP relaxation and the approximability threshold.

\paragraph{Fault-tolerant spanners.}

We also adapt our algorithms to the fault-tolerant setting, albeit restricted to unit length edges
(see Sections~\ref{sec:FT_main} and \ref{sec:stretch3_FT}).
The fault-tolerant setting is significantly more complicated:
the LP relaxation might have an exponential
number of both variables and constraints (see Section~\ref{sec:LP_FT}),
and we must resort to bicriteria approximations
when the number of faults $r$ is not constant.
Generally speaking, the approximation factors we obtain grow with $r$ like $k^r$ in the first algorithm (for general $k$),
and polynomially in $r$ in the second algorithm (for $k=3$ with unit-length edges).  These are the first results for fault-tolerant spanners in
\emph{directed} graphs.
For undirected graphs, absolute guarantees (i.e. not as approximation factors)
are known \cite{CLPR09}, and these bounds also grow like $k^r$
when vertices fail, even for the $k=3$ case.

\begin{table}
\begin{center}
    \begin{tabular}{|  l || l r | l r | l r |}
    \multicolumn{7}{l}{\bf Directed $k$-Spanner with unit edge-lengths} \\
    \hline
    Stretch
    & \multicolumn{2}{|c|}{Our Approximation}
    & \multicolumn{2}{|c|}{Previous Approximation}
    & \multicolumn{2}{|c|}{Integrality gap}
    \\ \hline\hline
    $k\ge 4$
    & $\tilde{O}(n^{2/3})$ & Thm \ref{thm:spanner_approx}
    & $\tilde{O}(n^{1-1/k})$ & \cite{BGJRW09}
    & $\Omega(\frac1k \cdot n^{1/3-\veps})$ & Thm \ref{thm:unit_integrality_gap}
    \\ \hline
    $k = 3$
    & $\tilde{O}(n^{1/2})$ & Thm \ref{thm:3spanner_thm}
    & $\tilde{O}(n^{2/3})$ & \cite{EP05,BGJRW09}
    & $\Omega(n^{1/3-\veps})$ & Thm \ref{thm:unit_integrality_gap}
    \\ \hline
    $k = 2$
    & $O(\log n)$ & Thm~\ref{thm:UBdirect1}
    & $O(\log n)$ & \cite{KP94,EP01}
    & $\Omega(\log n)$ & Thm~\ref{thm:2-spannerLB}
    \\
    \hline
  \end{tabular}
\end{center}

%\smallskip

\begin{center}
    \begin{tabular}{|  l || l r |  l r |}
    \multicolumn{5}{l}{\bf Similarly but with $r$ (vertex/edge) fault-tolerance}
    \\ \hline
    Stretch
    & \multicolumn{2}{|c|}{Our Approximation}
    & \multicolumn{2}{|c|}{Previous Approximation}
    \\ \hline\hline
    $k\ge 4$
    & $(\frac{1}{k(1+\epsilon)}, O((\frac{(1+\epsilon)r(k+r)^{k+r} n \ln n}{\epsilon r^r k^k})^{2/3}))$ & Thm \ref{thm:FT_approx}
    & \multicolumn{2}{|c|}{---}
    \\ \hline
    $k = 3$
    & $(\frac{1}{3(1+\epsilon)}, \tilde{O}(r n^{1/2}))$ & Thm \ref{thm:3spanner_FT}
    & \multicolumn{2}{|c|}{---}
    \\ \hline
    $k = 2$
    & $O(r \log n)$ & Thm \ref{thm:UBdirect1FT}
    & \multicolumn{2}{|c|}{---}
    \\ \hline
  \end{tabular}
\end{center}

\caption{Summary of our approximation results compared with previous work}
\label{tab:results}
\end{table}

\medskip

\paragraph{Note:} 
Shortly after completing the initial version of this paper we became
aware of a preprint by Berman, Raskhodnikova, and Ruan~\cite{BRR10},
that, independently of our work, constructs an
$O(k \cdot n^{1-1/\lceil k/2\rceil} \log n)$ approximation for
directed unit-length $k$-spanner.  For $k=3$ this gives essentially
the same bound as our $\tilde{O}(\sqrt{n})$-approximation (up to
logarithmic factors), and for $k=4$ their $O(\sqrt{n} \log
n)$ ratio is a polynomial improvement over our $O((n \log n)^{2/3})$
approximation.  For $k=5,6$ the approximation ratio of their algorithm
becomes $O(n^{2/3} \log n)$, basically matching ours.  Their
techniques are not based on linear programming; instead, they first
prove that every valid spanner can be covered by generalized stars using
few edges, and then design a set-cover-like approximation algorithm
for the minimum-size generalized-star cover problem.  While their
ratio is better than the previous
$\tilde{O}(n^{1-1/k})$-approximation~\cite{BGJRW09}, it still only
applies to the unit-length setting and degrades with $k$ so as to give
a nontrivial bound only when $k \leq O(\log n)$, while our
$\tilde{O}(n^{2/3})$-approximation suffers from neither of these
limitations.

\subsection{Techniques}

All of our approximation algorithms rely on solving the LP relaxation
and rounding the resulting ``fractional'' solution.
In some cases the LP relaxation does not have polynomial size,
but we can solve it within a reasonable approximation in polynomial time
by reducing it (via duality and the ellipsoid algorithm)
to a problem known in the literature as Restricted Shortest Path
(Theorems \ref{thm:LP_solve} and \ref{thm:FT_LP_solve}).
A virtue of all our algorithms is that they are relatively simple,
and thus can be extended to more complicated scenarios with little effort,
as is evident in the fault-tolerance case.

Our main technical contribution is to design new rounding procedures
which use few edges but are effective in creating many suitable paths.
Two very natural and well-known rounding techniques fail miserably:
(1) rounding separately each edge proportional to its LP value
(deterministically or randomly) is unlikely to form suitable paths,
and (2) rounding separately each flow-path
(say randomly \ala Raghavan and Thompson \cite{RaTh:RandomizedRounding})
will use far too many edges.
The challenge is thus to ``coordinate'' the selection of edges
so that they tend to create suitable paths.
Put differently, each $(u,v)\in E$
can be seen as a demand pair with its own flow,
and we need to select one flow-path for each pair
in a way that is ``biased'' towards using the same edges.
Ideally, we would like to select both the edges and the flow-paths
proportionally to their LP value.

Our first algorithm, for general $k$, is based on classifying demand pairs
according to whether they have ``few'' or ``many'' low-stretch paths.
The key insight is to make this classification rely on counting
\emph{vertices} participating in low-stretch paths.  The algorithm is then almost straightforward:
applying a threshold rounding of the LP handles pairs of the first type,
and building shortest-path arborescences from a small number of randomly
chosen vertices handles pairs of the second type.
This algorithm is described in Section \ref{sec:alg_main}.  To extend this algorithm to the fault-tolerant case
we use the appropriate LP relaxation and apply the
above rounding technique to several ``perturbations'' of the instance, each
obtained by deleting from the graph a random subset of vertices/edges.  This algorithm is described in Section~\ref{sec:FT_main}.
We believe that this perturbation technique, which we call failure sampling,
is of independent interest,
and may find future applications in related fault-tolerant problems.  For example, a simple application of failure sampling combined with existing \emph{undirected} spanner constructions gives a fault-tolerant undirected spanner construction with an absolute bound on its size that is only $\tilde{O}(r^3)$ larger than for a non-fault-tolerant spanner (where $r$ is the number of faults), while the best previous construction~\cite{CLPR09} has size $O(k^r)$ larger than a non-fault-tolerant spanner.

\medskip
Our second algorithm, for directed $3$-spanner with unit-length edges, uses randomized rounding
but at the level of \emph{vertices} rather than edges or flow-paths.  For every $v\in V$ we choose a random threshold $T_v\in[0,1]$,
and include in the solution every edge $(u,v)\in E$
for which, compared to the edge's LP value,
either $\minn{T_u,T_v}$ is ``small'' or $\maxx{T_u,T_v}$ is ``moderate''.
The probability of including an edge in this solution
is proportional to the edge's LP value,
but our conditions encourage positive correlation along a path
(e.g. for edges sharing an endpoint).
The proof boils down to considering a given demand pair,
and analyzing the possibly many different flow-paths between them.
In some cases, we control the correlation between these flow-paths
using Janson's inequality (see e.g. \cite{AS00,DP09}).
But in other cases such correlation analysis is not effective,
so we structurally ``decompose'' the paths into their
first, second and third hops and then use standard concentration bounds separately for each hop
plus some global arguments based on flow conservation.
This rounding procedure extends to the fault-tolerant case very easily;
we just repeat the rounding procedure several times with fresh coins.
These algorithms are described in Section~\ref{sec:stretch3}.

\subsection{Concluding Remarks and Future Work}\label{sec:conclusion}

Our results for directed $k$-spanner with unit edge-lengths address
what Elkin and Peleg~\cite{EP05} highlighted as two ``challenging directions'':
obtaining sublinear approximation for general $k$
and improving over their $\tO(n^{2/3})$-approximation for $k=3$.
LP-based approaches are quite generic yet often times optimal,
and thus it would not be surprising if the integrality gap of our LP relaxation
gives away the problem's true approximability threshold,
which appears to be polynomial with exponent in the range $[1/3,2/3]$
(possibly just $\tO(\sqrt n)$) for all $k\ge 3$.

A bolder conjecture would be that this LP relaxation also exposes
the approximability threshold for other spanner problems.
One such family of problems is the undirected setting, whose consideration we defer to future work but briefly discuss some preliminary results.
For $k=3$ the known approximation (and thus integrality gap)
is $O(\sqrt n)$ by the absolute guarantee of \cite{ADDJS93} (or our Theorem \ref{thm:3spanner_thm}),
while our best integrality is $\tOmega(n^{1/8})$,
which we prove by a somewhat tricky analysis
of the random graph $G_{n,p}$ for $p\approx 1/\sqrt n$.
For larger $k$, the known approximation (and thus integrality gap)
is $O(n^{2/(k+1)})$ via \cite{ADDJS93},
while our best integrality gap is $n^{\Omega(1/k)}$,
using the reduction designed by \cite{EP07}
from Min-Rep with large girth (which they conjecture to be NP-hard).
Compared to \cite{EP07}, our integrality gap may be viewed as
yet another indication to the inapproximability of basic $k$-spanners,
but with a factor much closer to the known approximation algorithm.

\section{Flow-Based LP Relaxation}\label{sec:LP}

We begin by describing the linear programming relaxation of the directed $k$-spanner problem that we will use.  Consider an instance of the directed $k$-spanner problem: a directed graph $G=(V,E)$ and an assignment of lengths to the edges $d: E \rightarrow \mathbb{R}^{+}$.  For $(u,v) \in E$, let $\mathcal P_{u,v}$ denote the set of all stretch $k$ paths (in $G$) from $u$ to $v$, i.e.~valid paths whose length is within factor $k$ of the shortest.  It is easy to see that the following LP is a relaxation of the $k$-spanner problem.  The variables are $x_e$, representing whether edge $e\in E$ is included in $G'$, and $f_P$, representing flow along path $P$.
\begin{equation} \label{LP:spanner1}
\framebox{ $
\begin{array}{lll}
  \min & \displaystyle \sum_{e\in E} x_e
  \\
  \mathrm{s.t.}
  & \displaystyle \sum_{P\in \mathcal P_{u,v}: e\in P} f_P \le x_e
  & \forall (u,v) \in E,\ \forall e\in E
  \\
  & \displaystyle \sum_{P\in \mathcal P_{u,v}} f_P \ge 1
  & \forall (u,v) \in E
  \\
  & \displaystyle x_e \geq 0 & \forall e \in E
  \\
  & \displaystyle f_P \geq 0 & \forall (u,v) \in E,\ \forall P \in \mathcal P_{u,v}
\end{array}
$ }
\end{equation}

In general LP~\eqref{LP:spanner1} has exponential size.  If the number of paths in $\mathcal P_{u,v}$ is at most polynomial for all $(u,v) \in E$ (for example if all lengths are unit and $k$ is a constant) then the LP obviously has only a polynomial number of variables and constraints, so can be solved optimally.  But in general we will need a different approach.   We will work with the dual, which has a polynomial number of variables and an exponential number of constraints.  So it is sufficient to find a separation oracle for the dual.  It turns out that the separation problem for the dual is the \emph{Restricted Shortest Path} problem, sometimes also called the \emph{Length Constrained Lightest Path} problem.  A PTAS is known for this problem~\cite{LR01, Has92}, so we can approximately separate for the dual.  This is enough to approximately solve the primal.

\begin{theorem} \label{thm:LP_solve}
There is an algorithm that in polynomial time computes a $(1+\epsilon)$ approximation to the optimal solution of LP~\eqref{LP:spanner1} for any constant $\epsilon > 0$.
\end{theorem}
\begin{proof}

The dual of LP~\eqref{LP:spanner1} is LP~\eqref{LP:spanner_dual}, which has a variable for every edge and a variable for every pair of edges.  The intuition is that the $y_{u,v}^e$ variables must form a ``fractional cut" of $(u,v)$ relative to stretch-$k$ paths.
\begin{equation} \label{LP:spanner_dual}
\framebox{ $
\begin{array}{lll}
  \max & \displaystyle \sum_{(u,v) \in E} z_{u,v} \\
  \mathrm{s.t.}
  & \displaystyle \sum_{(u,v) \in E} y_{u,v}^e \leq 1 & \forall e \in E \\
  & \displaystyle z_{u,v} - \sum_{e \in P} y_{u,v}^e \leq 0  & \forall (u,v)\in E,\ \forall P \in {\mathcal P}_{u,v} \\
  & \displaystyle z_{u,v} \geq 0 & \forall (u,v) \in E \\
  & \displaystyle y_{u,v}^e \geq 0 & \forall (u,v) \in E, \ \forall e \in E
\end{array}
$ }
\end{equation}

To construct a separation oracle for this LP, note there are only a
polynomial number ($|E|$) of constraints of the first type, so we can
just check them one by one. For constraints of the second type, note
that for every $u,v\in V$ the values $\{y_{u,v}^e\}_{e \in E}$ are
just a non-negative edge-weighting, and the constraint just require
all of the original stretch $k$ paths to have total length (under this
new weighting) of at least $z_{u,v}$.  So we get the following problem:
given two weightings of the same graph, find the shortest path under
the second weighting subject to having length at most $T>0$ under the
first weighting (for some threshold $T$).  If we could solve this we
would have a separation oracle for the dual.  Note that in the unit-length case stretch-$k$ paths correspond exactly to $k$-hop paths, so we can solve this problem exactly using Bellman-Ford.

For the general lengths setting, this problem has been considered in the literature under
the names ``Length Constrained Lightest Path'' and ``Restricted
Shortest Path''.  An FPTAS is known to exist \cite{LR01, Has92}, which
gives us an approximate separation oracle.  So by using Ellipsoid with this oracle we find a polynomial number of constraints such that the optimal solution violates all of the other constraints (which we did not include) by at most $1-\epsilon$, i.e.~there might be paths where $(1-\epsilon) z_{u,v} \leq \sum_{e \in P} y_{u,v}^e$.  So if we simply let $z'_{u,v} = (1-\epsilon) z_{u,v}$ we have a feasible solution for LP~\eqref{LP:spanner_dual} that is within $1-\epsilon$ of optimal.  So the optimum of LP~\eqref{LP:spanner_dual} is at least $1-\epsilon$ times the optimum of the compact dual (informally, changing to only a polynomial number of constraints did not affect the value of the optimal solution very much).  Thus by strong duality if we use solve a compact version of LP~\eqref{LP:spanner1} that has only the variables corresponding to the constraints found by Ellipsoid on the dual (of which there are only a polynomial number) we get a solution of value equal to the optimum of the compact dual, which is at most $1/(1-\epsilon)$ times the value of the actual dual (LP~\eqref{LP:spanner_dual}).
\end{proof}

\subsection{Fault Tolerant Relaxation} \label{sec:LP_FT}
There are two versions of the $r$-fault tolerant $k$-spanner problem, depending on whether we protect against edge faults or vertex faults. The idea in both cases is the same, so we shall focus on vertex faults: construct a subgraph $H$ of $G$ such that for every set $F$ of at most $r$ faulting vertices, $H \setminus F$ is a $k$-spanner of $G \setminus F$.  We can change LP~\eqref{LP:spanner1} to support this version by allowing a different set of flows $\{f_P^F\}$ for every such $F$, but using the same capacity variables $\{x_e\}$.  More formally, let $\mathcal{P}_{u,v}^F$ be the set of stretch-$k$ paths from $u$ to $v$ in $G \setminus F$ (where $F$ is a set of edges for the edge tolerant version or is a set of vertices for the vertex tolerant version).  For  $F \subseteq V$, let $E_F \subseteq E$ be the set of edges with at least one endpoint in $F$.  We will use the following relaxation for the vertex version (the edge version is analogous):
\begin{equation} \label{LP:FT_spanner}
\framebox{ $
\begin{array}{lll}
  \min & \displaystyle \sum_{e \in E} x_e
  \\
  \mathrm{s.t.}
  & \displaystyle \sum_{P \in \mathcal P_{u,v}^F:\ e \in P} f_P^F \leq
  x_e & \forall F \subseteq V : |F| \leq r,\
  \forall (u,v) \in E \setminus E_F,\
  \forall e \in E \setminus E_F
  \\
  & \displaystyle \sum_{P \in \mathcal P_{u,v}^F} f_P^F \geq 1
  & \forall F \subseteq V : |F| \leq r,\ \forall (u,v)\in E \setminus E_F
  \\
  & \displaystyle x_e \geq 0 & \forall  e \in E
  \\
  & \displaystyle f_P^F \geq 0 & \forall F \subseteq V : |F| \leq r, \ \forall (u,v) \in E \setminus E_F, \ \forall P \in \mathcal P_{u,v}^F
\end{array}
$ }
\end{equation}

It is easy to see that LP~\eqref{LP:FT_spanner} has $n^{O(r)}$ constraints.  Each possible fault set acts like an instance of the original spanner LP~\eqref{LP:spanner1}, except for the sharing of the capacity variables $\{x_e\}_{e \in E}$.  So when we take the dual we get a program with $n^{O(r)}$ variables, and if $r$ is constant the separation oracle we designed for the non-fault-tolerant version suffices to separate this LP as well.  When $r$ is not constant this technique does not work as the dual will have a superpolynomial number of variables.  Instead we will give a bicriteria algorithm, which in the unit-length case will find a fractional assignment to the $x_e$ variables of cost at most $\frac{1+\epsilon}{\epsilon}$ times larger than the cost of the best fractional solution, and that supports flows that satisfy the constraints for all $F$ of size at most $r / ((1+\epsilon) k)$.

\begin{theorem} \label{thm:FT_LP_solve}
 For any $\epsilon >  0$ there is a polynomial time algorithm that, given an instance of the unit-length directed $r$-fault-tolerant $k$-spanner problem, finds a set of fractional capacities $\{x_e\}_{e \in E}$ with the following two properties: 1) there exist flow variables that satisfy the flow and capacity constraints of LP~\eqref{LP:FT_spanner} for all fault sets of size at most $\frac{r}{(1+\epsilon)k}$, and 2) $\sum_{e \in E} x_e$ is at most $\frac{1+\epsilon}{\epsilon}$ times larger than the optimal solution to LP~\eqref{LP:FT_spanner}
\end{theorem}
\begin{proof}
When $r$ is super-constant there is a super-polynomial number of constraints in the LP, so we cannot solve it using earlier methods (when we transform to the dual we get a super-polynomial number of variables).  Instead of going through the dual we will stick with the primal and give a separation oracle.  However, since ellipsoid with a separation oracle takes time polynomial in the dimension (i.e.~the number of variables) we  need to transform the problem into one with a polynomial number of variables.  We do this in a simple way: we simply project the polytope down on the capacity variables $x_e$, of which there are only $O(m)$.  The objective function of LP~\eqref{LP:FT_spanner} uses only the $x_e$ variables, so optimizing over this projection is sufficient to optimize over the full LP.  And since this is a projection of a convex set it is itself convex, so if we can design a separation oracle the ellipsoid algorithm will run in polynomial time.

So what would a separation oracle for this projected polytope be?  Simply examining LP~\eqref{LP:FT_spanner} shows that a setting of the capacity variables $\{x_e\}_{e \in E}$ is not a valid solution if and only if there is some set of at most $r$ faults such that it is impossible to send $1$ unit of flow between all demands.  Slightly more formally, $\{x_e\}_{e \in E}$ is not a valid solution if and only if there is some fault set $F$ (of size at most $r$) and edge $(u,v) \in E$ such that:
 \begin{enumerate}
    \item $(u,v) \not\in F$ for the edge-fault case, or $u,v \not\in F$ for the vertex-fault case, and
    \item There is no way of sending one unit of flow along stretch-$k$ paths from $u$ to $v$ in $G\setminus F$ while respecting capacities $\{x_e\}$.
 \end{enumerate}

 By strong duality, the maximum flow that can be sent along stretch-$k$ $u-v$ paths is equal to the smallest fractional cut, where a fractional cut is an assignment of values $y_e$ to the edges such that $\sum_{e \in P} y_e \geq 1$ for all $P \in \mathcal P_{u,v}^F$.  The size of such a cut for a particular fault set $F$ is $\sum_{e \in E} x_e y_e$.  So for every fault set $F$ of size at most $r$, for every remaining demand $(u,v)$, for every fractional cut $\{y_e\}_{e \in E}$ relative to $F$ and to $(u,v)$, any feasible solution $\{x_e\}_{e \in E}$ has $\sum_{e \in E} y_e x_e \geq 1$.  These are the violated constraints that our separation oracle will find (or at least will approximately find).

So to construct a separation oracle, we want to find a set of faults $F$ and demand $(u,v)$ with the smallest fractional cut.  If the size of this cut is less than $1$, then we have found a separating hyperplane, and if there is no such set $F$ then the current capacities are feasible.  In order to solve this problem, which we will call {\sc Stretch-$k$ Interdiction}, we first write it as a mixed-integer program.  Since there are only a polynomial number of $(u,v)$ demands we can simply try them all, so our formulation is for some given $(u,v)$.  This formulation is for the vertex-fault version; the edge-fault version follows the same basic idea. Recall that the $x_e$'s are the capacity variables in the original problem, so in this context they are fixed constants and thus the objective function and the constraints are linear.  The intention of MIP~\eqref{MIP:interdiction} is for $z_w$ to represent whether vertex $w$ is part of the fault set and for the $\{y_e\}$ variables to represent a fractional cut of the remaining paths.

\begin{equation} \label{MIP:interdiction}
\framebox{ $
\begin{array}{lll}
  \min & \displaystyle \sum_{e \in E} x_e y_e
  \\
  \mathrm{s.t.}
  & \displaystyle \sum_{(a,b) \in P} (y_{(a,b)} + \frac12 z_a + \frac12 z_b) \geq 1 & \forall P \in \mathcal P_{u,v}
  \\
  & \displaystyle \sum_{w \in V} z_w \leq r
  \\
  & \displaystyle z_u = z_v = 0
  \\
  & \displaystyle z_w \in \{0,1\} & \forall w \in V
  \\
  & \displaystyle y_e \geq 0 & \forall e \in E
\end{array}
$ }
\end{equation}

\begin{claim}
MIP~\eqref{MIP:interdiction} is an exact formulation of {\sc Stretch-$k$ Interdiction}
\end{claim}
\begin{proof}
Note that there is a one-to-one correspondence between settings of the $z_w$ variables and possible fault sets.  For every setting of the $z_w$'s, the objective value is the minimum cost fractional cut (where we have to cut stretch-$k$ paths that do not hit any faults), which is exactly what we are trying to optimize.
\end{proof}

\begin{lemma}
There is a bicriteria approximation for {\sc Stretch-$k$ Interdiction} that uses at most $(1+\epsilon) k r$ faults (instead of $r$) and has cost at most $\frac{1+\epsilon}{\epsilon}$ times the best $r$-fault solution
\end{lemma}
\begin{proof}
In order to solve MIP~\eqref{MIP:interdiction} we relax the integrality constraints on the $z_w$ variables to $0 \leq z_w \leq 1$, giving us a linear program.  We can solve the resulting LP by constructing its own separation oracle: if we define the length of an edge $(a,b)$ to be $y_{(a,b)} + \frac12 z_a + \frac12 z_b$, is there a stretch-$k$ path with length less than $1$?  Since we only consider the fault-tolerant setting for the unit-length case we can actually solve this problem exactly using Bellman-Ford (since in this case stretch-$k$ is equivalent to $k$-hop).  So we can solve this LP in polynomial time.

Now we need to round the $z_w$ variables to integers.  We will use a very simple threshold rounding: if $z_w \geq \frac{1}{(1+\epsilon)k}$ then set $z'_w = 1$; otherwise set $z'_w = 0$.  Furthermore, set $y'_e = \frac{1+\epsilon}{\epsilon} y_e$.  Since any stretch-$k$ path is a $k$-hop path, if $\sum_{w \in P} z_w \geq \frac{1}{1+\epsilon}$ then $z_w \geq \frac{1}{(1+\epsilon) k}$ for some $w \in V$, and thus $z'_w$ covers $P$.  On the other hand, if $\sum_{w \in P} z_w < \frac{1}{1+\epsilon}$ then $\sum_{e \in P} y_e > \frac{\epsilon}{1+\epsilon}$, so $\sum_{e \in P} y'_e = \frac{1+\epsilon}{\epsilon} \sum_{e \in P} y_e \geq 1$ and the $y'$ variables cover $P$.  Thus $(z', y')$ is a valid solution to the MIP except that $\sum_{w \in V} z'_w \leq (1+\epsilon)k r$ instead of being at most $r$.  In other words, we have designed a $((1+\epsilon)k, \frac{1+\epsilon}{\epsilon})$-bicriteria approximation algorithm for MIP~\eqref{MIP:interdiction} and thus for {\sc Stretch-$k$ Interdiction}.
\end{proof}

By using this bicriteria approximation with original fault budget $r / (1+\epsilon)k$ instead of $r$, we will find a separating hyperplane (whose coefficients are the $\{y'_e\}$ variables) as long as there is some fault set $F$ of size at most $r / (1+\epsilon)k$ and demand $(u,v)$ for which the maximum stretch-$k$ flow (or equivalently the minimum fractional cut) is at most $\epsilon / (1+\epsilon)$.  So using this separation oracle with the Ellipsoid algorithm and then rounding the capacities we find up by $\frac{1+\epsilon}{\epsilon}$ gives us a bicriteria algorithm for LP~\eqref{LP:FT_spanner}, yielding the theorem.
\end{proof}

\section{Approximations for Directed $k$-Spanner} \label{sec:alg_main}

We will now design a $\tilde{O}(n^{2/3})$-approximation
algorithm for the directed $k$-spanner problem.  We first solve LP~\eqref{LP:spanner1} as detailed in Theorem~\ref{thm:LP_solve} to get a fractional solution $(x,f)$.  We then round this solution using Algorithm~\ref{alg:round}, which has two main components: a simple threshold rounding scheme together with a collection of shortest path arborescences.

\begin{algorithm}[]
\caption{Rounding Algorithm for Directed $k$-spanner}
\label{alg:round}
$E' \leftarrow \{e \in E : x_e \geq 1/(3n \ln n)^{2/3}\}$
\\
\For{$i \leftarrow 1$ \KwTo $(3n \ln n)^{2/3}$} {
Choose $v \in V$ uniformly at random
\\
$T_i^{in} \leftarrow$ shortest path in-arborescence rooted at $v$
\\
$T_i^{out} \leftarrow$ shortest path out-arborescence rooted at $v$
}
Output $E' \bigcup \left(\cup_{i=1}^{(3n \ln n)^{2/3}} (T_i^{in} \cup T_i^{out})\right)$
\end{algorithm}

To show that this algorithm gives a valid $k$-spanner, we begin with a lemma that characterizes edges that are satisfied by the thresholding.  For every $(u,v)$ in $E$, let $N_{u,v} \subseteq V$ be the set of vertices that lie on a path of stretch at most $k$ from $u$ to $v$ (i.e.~the set of vertices that are used by at least one path in $\mathcal{P}_{u,v}$).

\begin{lemma}\label{lem:spanner_path}
For any $(u,v) \in E$ there is a path $P \in \mathcal{P}_{u,v}$ with
the property that every edge $e \in P$ has $x_e \geq \frac{1}{|N_{u,v}|^2}$
\end{lemma}
\begin{proof}
  Suppose this is false for some $(u,v)$.  Let $B \subseteq N_{u,v}
  \times N_{u,v}$ be the set of edges with $x_e < 1/|N_{u,v}|^2$.
  Then every path $P \in \mathcal P_{u,v}$ goes through at least one
  edge in $B$, so these edges form a cut between $u$ and $v$ relative
  to the paths in $\mathcal P_{u,v}$.  Since we have a valid LP
  solution, we know that at least one unit of flow is sent from $u$ to
  $v$ using paths in $\mathcal P_{u,v}$.  This means that the number
  of edges in $B$ must be at least $|N_{u,v}|^2$.  But this is a
  contradiction: every edge in $B$ has both endpoints in $N_{u,v}$, so
  there are at most ${|N{u,v}| \choose 2} < |N_{u,v}|^2$ of them.
\end{proof}

So if $|N_{u,v}|$ is small, Lemma~\ref{lem:spanner_path} implies that
there is some stretch $k$ path with the property that every edge is
assigned a large capacity.  On the other hand, if $|N_{u,v}|$ is large then there are many nodes that are
on stretch $k$ paths, so we should be able to find such a path by
picking nodes randomly.  This is formalized in the following lemma:

\begin{lemma}\label{lem:spanner_sample}
  If we sample at least $\frac{3n\ln n}{|N_{u,v}|}$ vertices
  independently and uniformly at random, then with probability at
  least $1-1/n^3$ at least one sampled vertex will be in $N_{u,v}$
\end{lemma}
\begin{proof}
  The probability that no sampled vertex is in $N_{u,v}$ is at most $\left(1-\frac{|N_{u,v}|}{n}\right)^{\frac{3n\ln n}{|N_{u,v}|}} \leq e^{-3\ln n} = 1/n^3$ and thus the probability that at least one sampled vertex is in
  $N_{u,v}$ is at least $1-1/n^3$
\end{proof}

\begin{theorem}\label{thm:spanner_approx}
  There is a polynomial time algorithm that with high probability returns a directed $k$-spanner of size at most $O((n \ln n)^{2/3})$ times the smallest directed $k$-spanner, for any $k \geq 1$.
\end{theorem}
\begin{proof}
  The algorithm is simply to solve LP~\eqref{LP:spanner1} using Theorem~\ref{thm:LP_solve}, and then round the solution using Algorithm~\ref{alg:round}.  We first prove that it results in a valid spanner with high probability.  Consider some
  edge $(u,v) \in E$.  If $|N_{u,v}| \leq (3 n \ln n)^{1/3}$, then
  Lemma~\ref{lem:spanner_path} implies that there is some stretch $k$
  path from $u$ to $v$ using edges contained in $E'$ and thus in the spanner.  On the other hand, if
  $|N_{u,v}| > (3 n \ln n)^{1/3}$ then Lemma~\ref{lem:spanner_sample}
  implies that with probability at least $1-1/n^3$ we will have
  sampled some vertex in $N_{u,v}$.  Suppose we sample $w \in
  N_{u,v}$ on the $i$th iteration.  By the definition of $N_{u,v}$ we know that $w$ is on
  some path from $u$ to $v$ with stretch at most $k$, and thus the
  length of the shortest path from $u$ to $w$ plus the length of the
  shortest path from $w$ to $v$ is at most $k \cdot d_G(u,v)$.  These paths are contained in $T_i^{in} \cup T_i^{out}$, so the spanner will include both of
  these shortest paths and thus will include a path from $u$ to $v$
  with stretch at most $k$.  Taking a union bound over all $(u,v)$ completes the proof that the returned subgraph is a $k$-spanner.

  To prove that it is a $O((n \ln n)^{2/3})$-approximation we will
  show that each of the two steps costs at most $O((n \ln n)^{2/3}) \times OPT$.
  This is obvious for the LP rounding step: an edge $e$ is in $E'$ only if $x_e \geq 1/(3 n\ln n)^{2/3}$, so $|E'|$ is at most  $O((n \ln n)^{2/3})$ times the LP cost.  To
  show that the second step does not add many edges, note that every iteration adds at most $2(n-1)$ edges, and that $n-1$ is a trivial lower bound on $OPT$ (since we are assuming the underlying graph is connected; if it is not connected then it is easy to modify this analysis to still hold).  Thus the total
  cost of all the arborescences is at most $2(3 n \ln n)^{2/3} \times OPT$
\end{proof}

\subsection{Extension to $r$-Fault-Tolerant Version} \label{sec:FT_main}

In order to adapt the rounding scheme of Algorithm~\ref{alg:round} to the fault-tolerant case we need to show how to modify the threshold rounding and the arborescence sampling.  It is simple to see that Lemma~\ref{lem:spanner_path} still holds for every fault set, so the threshold rounding will still work (although we will change the threshold).  But the arborescence rounding must be changed to allow for faults.  The intuition behind the change comes from the technique of \emph{color-coding}~\cite{AYZ95}: before randomly sampling the root of an arborescence, independently fail each element (either edges in the edge-tolerant version or vertices in the vertex-tolerant version) with some probability $p$.
We call this technique \emph{failure sampling}.  We then randomly sample a root and include its shortest-path in- and out-arborescences in the resulting subgraph.  By setting $p$, the number of arborescences sampled, and the threshold of the rounding appropriately, we get the following theorem.  We say that an algorithm is an $(\alpha, \beta)$-approximation for the $r$-fault-tolerant directed $k$-spanner problem if it returns a $\alpha r$-fault-tolerant $k$-spanner of size at most $\beta$ times the smallest $k$-spanner, and an algorithm is a \emph{true} $\beta$-approximation if it is a $(1,\beta)$-approximation.

\begin{theorem} \label{thm:FT_approx}
For any constant $\epsilon > 0$ there is a polynomial-time algorithm that is a \newline $\left(\frac{1}{(1+\epsilon)k}, O\left(\left(\frac{(1+\epsilon)r(k+r)^{k+r} n \ln n}{\epsilon r^r k^k}\right)^{2/3}\right)\right)$-approximation for the unit-length $r$-fault-tolerant directed $k$-spanner problem. There is also a true approximation algorithm with the same approximation ratio that takes $n^{O(r)}$ time.
\end{theorem}
\begin{proof}
Suppose that we have a feasible solution for LP~\eqref{LP:FT_spanner} (or at least an approximate solution from Theorem~\ref{thm:FT_LP_solve}).  We show how to round a fractional solution into an integer solution, assuming that all lengths are $1$.  Our rounding is basically the same as for the non-fault tolerant version.  In particular, Lemma~\ref{lem:spanner_path} still holds for every fault set.  So we can, as before, set a threshold value $t$ and round up any edge with $x_e \geq 1/t$.  The only difference comes in the random sampling step: in the non-fault-tolerant version, it sufficed to randomly pick centers of shortest path in- and out-arborescences.  But in the fault-tolerant setting that is no longer sufficient; the paths we construct must suffice even after failures, which simple shortest paths obviously will not.  So we will add an extra step inspired by \emph{color-coding}~\cite{AYZ95}: before randomly sampling the root of an arborescence, independently fail each element (either edges in the edge-tolerant version or vertices in the vertex-tolerant version) with some probability $p$.  We call this technique \emph{failure sampling}.  We then randomly sample the root of shortest-path in- and out-arborescences in the resulting subgraph.

Consider some fault set $F$ and some edge $(u,v)$ that still survives in $G \setminus F$.  Define $N_F(u,v)$ as in the no-fault setting: a vertex $x$ is in $N_F(u,v)$ if $x$ is on some stretch-$k$ $u-v$ path in $G \setminus F$.  Since Lemma~\ref{lem:spanner_path} still holds, if $|N(u,v)| \leq \sqrt{t}$ then the threshold rounding satisfies the demand.  So we assume that $|N_F(u,v)| > \sqrt{t}$ and analyze the probability that a single round of the random sampling will satisfy the demand.  We will then perform the number of rounds necessary to be able to take a union bound over all possible $F$ and $(u,v)$.

Since we are assuming all edge lengths are $1$, a stretch-$k$ path is the same as a $k$-hop path.  A sufficient condition for the sampling to succeed for $F$ and $(u,v)$ is for the arborescence root to be a vertex in $N_F(u,v)$, everything in $F$ to be killed by the failure sampling, and nothing from the $k$-hop path containing the root to be killed (note that such a path must exist by the definition of $N_F(u,v)$).  Given that the the root is selected to be in $N_F(u,v)$, the probability that the particular $k$-hop $u-v$ path containing the root is all preserved by the failure sampling is $(1-p)^k$.  And clearly the probability that everything from $F$ is killed by the failure sampling is $p^k$, and is independent of the other two events.  So the probability that all three events happen is $p^r \times \frac{|N_F(u,v)|}{n} \times (1-p)^k \geq p^{r} (1-p)^k \frac{\sqrt{t}}{n}$.  The number of possible failure sets $F$ and demands $(u,v)$ is at most ${m \choose r} \times {n \choose 2} \leq n^{2r + 2}$ in the edge-failure setting; for vertex failures it is at most $n^{r+2}$.  Let $\ell$ be the number of rounds for which we repeat the random sampling.  Then in order to succeed on all constraints with probability at least $1/2$, we want
\begin{equation*}
\left(1 - \frac{p^{r} (1-p)^k \sqrt{t}}{n}\right)^{\ell} \leq \frac{1}{2n^{2r+2}}
\end{equation*}

Setting $p = \frac{r}{k+r}$ and solving for $\ell$, we get that it is sufficient to set
\begin{equation*}
\ell = \frac{2(2r+2)(k+r)^{k+r} n \ln n}{r^r k^k \sqrt{t}}
\end{equation*}

As in the non-fault-tolerant case, we balance out the cost of the sampling ($\ell$) with the cost of the threshold rounding ($t$) to get a total approximation of
\begin{equation*}
O\left(\left(\frac{r(k+r)^{k+r} n \ln n}{r^r k^k}\right)^{2/3}\right)
\end{equation*}

If $r$ is constant, then this rounding combined with our ability to actually solve LP~\eqref{LP:FT_spanner} gives us the theorem for the $r=O(1)$ case.  If $r$ is not constant then we need to use Theorem~\ref{thm:FT_LP_solve} before the rounding procedure, giving us the claimed bicriteria approximation.\end{proof}

\subsection{Integrality Gap} \label{sec:integrality_gap}
We now complement our approximation algorithm by proving a nearly matching integrality gap.  We do this by a reduction from the {\sc Min-Rep} problem.  In {\sc Min-Rep} we are given a bipartite graph $G = (U,V,E)$ together with a partition of $U$ and $V$ into groups $U_1, U_2 \dots, U_{p}$ and $V_1, V_2 ,\dots V_{p}$.  We say that there is a \emph{superedge} between two groups $(U_i, V_j)$ if there is some $u \in U_i$ and some $v \in V_j$ such that $\{u,v\} \in E$.  The goal is to find a subset $X \subseteq U \cup V$ of as few vertices as possible such that for every pair of groups $(U_i, V_j)$ with a superedge there is some $u \in U_i \cap X$ and $v \in V_j \cap X$ such that $\{u,v\} \in E$.

Elkin and Peleg~\cite{EP07} proved hardness for directed $k$-spanner by using a reduction from {\sc Min-Rep}, and we will use their reduction to prove an integrality gap.  But instead of reducing from generic {\sc Min-Rep} instances as in a hardness proof, we will only apply the reduction to instances of {\sc Min-Rep} in which every superedge actually corresponds to a matching between vertices, i.e.~if $(U_i, V_j)$ is a superedge then there is a matching between $U_i$ and $V_j$.  The interested reader might note that these are basically instances of the \emph{Unique Games Problem}~\cite{Khot02}.

We first give a lemma that was proved implicitly by Charikar, Hajiaghayi, and Karloff~\cite{CHK09}:

\begin{lemma} \label{lem:CHK_UG_instance}
For any constant $\epsilon > 0$, there are instances of {\sc Min-Rep} with the following properties:
\begin{enumerate}
\item Every group has size $n^{\frac23 - \epsilon}$
\item $OPT \geq n^{\frac23 - \epsilon}$
\item There is a matching between every $U_i$ and $V_j$
\end{enumerate}
\end{lemma}

We now use this lemma to prove the main integrality gap theorem by applying the Elkin and Peleg reduction~\cite{EP07} to instances from Lemma~\ref{lem:CHK_UG_instance} and showing there it has a small fractional solution.

\begin{theorem} \label{thm:unit_integrality_gap}
The integrality gap of LP~\eqref{LP:spanner1} on the unit-length directed $k$-spanner problem is $\Omega(\frac1k n^{\frac13 - \epsilon})$ for any constant $\epsilon > 0$.
\end{theorem}
\begin{proof}
The instances that we use to prove this integrality gap are the instances we obtain by applying the reduction of Elkin and Peleg~\cite{EP07} to the instances from Lemma~\ref{lem:CHK_UG_instance}.  We explain the reduction in detail so as to analyze the best fractional LP solution.  Let $r$ be the number of groups, so in the instances from Lemma~\ref{lem:CHK_UG_instance} we have that $r = n^{\frac13 + \epsilon}$.  For each group we will add $x = n^{\frac23 - \epsilon} / ((k-1)/2)$ paths, where each path has length $(k-1)/2$ (for ease of exposition we assume that $k$ is odd, but it does not actually matter).  More formally, let $(U', V', E')$ be a {\sc Min-Rep} instance from Lemma~\ref{lem:CHK_UG_instance}, and let $n_{MR} = |U' \cup V'|$.  Then the vertex set of our spanner problem is
\begin{equation*}
V = \left(U' \cup V'\right) \bigcup \left(\cup_{p=1}^x \cup_{i=1}^{r/2} \cup_{j=1}^{(k-1)/2} s_{i,j}^p \right) \bigcup \left(\cup_{p=1}^x \cup_{i=1}^{r/2} \cup_{j=1}^{(k-1)/2} t_{i,j}^p \right)
\end{equation*}

Note that $n = |V| = n_{MR} + x\cdot r \cdot (k-1)/2 = n_{MR} + ((n_{MR}^{\frac23 - \epsilon} / ((k-1)/2)) \cdot n_{MR}^{\frac13 + \epsilon} \cdot ((k-1)/2) = 2n_{MR}$, so we have only doubled the number of vertices.

The edge set is divided into a few different components.  First, let $E'' = \{(u,v) : u \in U', v \in V', \{u,v\} \in E'\}$ be the original {\sc Min-Rep} edges but now directed from $U$ to $V$.  Next we add a clique to every group: let $E_C = \cup_{i=1}^{r/2} ((U_i \times U_i) \cup (V_i \times V_i))$.  We also want to turn the new vertices into paths: let
\begin{equation*}
E_M = \bigcup_{p=1}^x \bigcup_{i=1}^{r/2} \bigcup_{j=1}^{(k-3)/2} \{(s_{i,j}^x, s_{i,j+1}^x), (t_{i,j}^x, t_{i,j+1}^x)\}
\end{equation*}
so for every fixed $i$ and $p$, the $E_M$ edges form a directed path from $s_{i,1}^p$ to $s_{i,(k-1)/2}^x$, and similarly for the $t$ vertices.  We also add edges to connect these paths to the original vertices: let
\begin{equation*}
E_U = \bigcup_{p=1}^x \bigcup_{i=1}^{r/2} \left(\{(s_{i, (k-1)/2}^p, u) : u \in U_i\} \cup \{(v, t_{i,1}^p) : v \in V_i\}\right)
\end{equation*}

Finally, we want to add edges to connect the endpoints of paths corresponding to superedges (which in our case is all $(U_i, V_j)$: let
\begin{equation*}
E_I = \bigcup_{p=1}^x \bigcup_{i=1}^{r/2} \bigcup_{j=1}^{r/2} \{(s_{i,1}^p, t_{j, (k-1)/2}^p)\}
\end{equation*}.

So our final edge set is $E = E'' \cup E_C \cup E_M \cup E_U \cup E_I$.

Elkin and Peleg~\cite{EP07} showed that the optimal spanner has size at least $\Omega(x \times OPT_{MR})$, where $OPT_{MR}$ is the size of the smallest {\sc Min-Rep} solution.  So in our case, the best spanner has size at least
\begin{equation*}
\Omega\left(\frac{n_{MR}^{\frac23 - \epsilon}}{(k-1) /2} \times n_{MR}^{\frac23 - \epsilon}\right) = \Omega\left(\frac1k n^{\frac43 - 2\epsilon}\right)
\end{equation*}

On the other hand, we claim that the best fractional solution is small, namely $O(n)$.  To see this, consider the following fractional assignment.  All edges inside the tails, i.e.~all edges in $E_M$, have fractional capacity $1$.  Let $q = n_{MR} / r = n^{\frac23 - \epsilon}$ be the size of each group.  We set the fractional capacity of all edges in $E'' \cup E_C \cup E_U$ to $2/q$, and set the fractional capacity of all edge in $E_I$ to $0$.  The cost of this solution is
\begin{align*}
|E_M| + \frac2q \left(|E''| + |E_C| + |E_U|\right) &\leq xr\frac{k-3}{2} + \frac2q \left(\frac{r^2 q}{4} + q^2 r + xrq\right) \\
& \leq xr\frac{k-3}{2} + \frac{2r^2}{4} + 2qr + 2xr \\
& \leq n_{MR} + n_{MR}^{\frac13 + \epsilon} + 2n_{MR} + 2n_{MR} \\
& = O(n)
\end{align*}

So it remains to prove that it is a valid fractional solution.  We proceed by analyzing each type of edge.  Obvious since every edge in $E_M$ is included with capacity $1$, we are able to send one unit of flow.  For some edge $(u,v) \in E''$, we can send $q-1$ flows, each of size $1/q$, first to the vertices in the same group as $u$ (via the $E_C$ edges), then across the matching to the group containing $v$ (via $E''$ edges), and then back to $v$ (via $E_C$ edges).  We can send the final $1/q$ flow directly on the edge $(u,v)$.  These paths have length $3 \leq k$ and obviously satisfy capacity constraints.  For some edge $(a,b) \in E_C$, we can do the same thing without crossing any matching: send $1/(q-1) \leq 2/q$ flow to each of the group mates of $a$, and then back into $b$.  For an edge $(s_{i, (k-1)/2}^p, u) \in E_U$, we can send $1/q$ flow from $s_{i, (k-1)/2}$ to each of the vertices in the group containing $u$ (using $E_U$ edges) and the from those vertices to $u$ (using $E_C$ edges).  Similarly, for an edge $(v, t_{i,1}^p)$ we can send $1/(q-1) < 2/q$ flow to each of the other vertices in $V_i$, and then from these vertices to $t_{i,q}^p$.  Finally, for an edge $(s_{i,1}^p, t_{j, (k-1)/2}^p) \in E_I$ we can send one unit along the $s_{i, \dots}^p$ path, then split it into $q$ paths using the $E_U$ edges, send each of those $1/q$ flows across the matching using $E''$ edges, recombine the flow at $t_{j, 1}^p$ using $E_U$ edges, and then send it down the path to $t_{j, (k-1)/2}^p$.  Each of these $q$ paths has length $k$ and satisfies the fractional capacities.  Thus this is a valid fractional solution.
\end{proof}

The reason that we lose an extra $\frac1k$ in the integrality gap is that because we are in the unit-length case we need to add many vertices in order to build long paths.  It is easy to modify the proof of Theorem~\ref{thm:unit_integrality_gap} to give a gap of $\Omega(n^{\frac13 - \epsilon})$ in the arbitrary lengths setting.

\section{Directed Unit-Length $3$-Spanner} \label{sec:stretch3}
While our $\tilde{O}(n^{2/3})$-approximation is an improvement over previous work for arbitrary edge lengths and for unit edge lengths with $k > 3$, for unit edge lengths with $k=3$ it matches the previous bounds of Elkin and Peleg~\cite{EP05} and Bhattacharyya et al.~\cite{BGJRW09}.  So for the specific case of unit-length directed $3$-spanner we develop a different rounding algorithm for the flow-based LP~\eqref{LP:spanner1} that gives an $\tilde{O}(\sqrt{n})$-approximation.  Our algorithm first solves LP~\eqref{LP:spanner1} and then rounds it using Algorithm~\ref{alg:3spanner}.  Informally, this rounding works by choosing a threshold value $T_v \in [0,1]$ for each vertex $v \in V$.  We then add all edges $(u,v)$ where either $T_u$ or $T_v$ is at most $\rho x_{u,v}$, where $\rho = \Theta(\sqrt{n} \log n)$ is an inflation factor to make the probabilities large enough.  This turns out to not be quite enough edges, so we also add all edges $(u,v)$ where \emph{both} $T_u$ and $T_v$ are at most $\sqrt{\rho x_{u,v}}$.  For technical reasons we have to add an extra complication: every vertex will actually choose another threshold, $T'_v$, and edges are added as described for every combination of $T$ and $T'$ thresholds.

\begin{algorithm}[H]
\caption{Rounding Algorithm for $3$-spanner}
\label{alg:3spanner}
Set $\rho = C \sqrt{n} \log n$ for a large constant $C$
\\
For every $v \in V$ choose independently two values $T_v, \ T'_v \in_R [0,1]$
\\
Let $E_1 = \{(u,v) \in E : \min\{T_u, T'_u, T_v, T'_v\} \leq \rho \cdot x_{u,v}\}$
\\
Let $E_2 = \{(u,v) \in E : \max\{\min\{T_u, T'_u\}, \min\{T_v, T'_v\}\} \leq \sqrt{\rho \cdot x_{u,v}}$
\\
Output $E' = E_1 \cup E_2$
\end{algorithm}

\begin{lemma} \label{lem:3spanner_cost}
Algorithm~\ref{alg:3spanner} returns a set of edges $E'$ with $\EX[|E'|] \leq O(\rho)$ times the size of the smallest $3$-spanner.
\end{lemma}
\begin{proof}
Let $(u,v) \in E$.  Obviously $\Pr[T_u \leq \rho x_{u,v}] \leq \rho x_{u,v}$, and similarly for $T'_u, T_v,$ and $T'_v$. So $\Pr[(u,v) \in E_1] \leq 4\rho x_{u,v}$ by a simple union bound.  To analyze $E_2$, note that the probability that $(u,v) \in E_2$ is equal to the probability that $\min\{T_u, T'_u\} \leq \sqrt{\rho x_{u,v}}$ and $\min\{T_v, T'_v\} \leq \sqrt{\rho x_{u,v}}$.  Since these are independent, the probability that the both happen is equal to the product of their probabilities.  And by another union bound we get that each of the probabilities is at most $2 \sqrt{\rho x_{u,v}}$, and thus the probability that $(u,v) \in E_2$ is at most $4 \rho x_{u,v}$.  Thus the probability that $(u,v) \in E'$ is at most $8 \rho x_{u,v}$, so by linearity of expectations the expected number of edges in $E'$ is at most $8 \rho \sum_{(u,v) \in E} x_{u,v}$, which is exactly $8\rho$ times the cost of the LP solution and thus at most $8 \rho$ times the size of the smallest $3$-spanner.\end{proof}

\begin{lemma} \label{lem:3spanner_correct}
For every edge $(u,v) \in E$, the probability that there is no path of length at most $3$  from $u$ to $v$ in $E'$ is at most $1/e$
\end{lemma}
\begin{proof}
We know that in the LP solution a flow of $1$ unit was sent from $u$ to $v$ along paths of length at most $3$.  If at least $1/3$ flow was sent from $u$ to $v$ through paths of length $1$ (i.e.~through $(u,v)$), then $x_{u,v} \geq 1/3$ so $\rho x_{u,v} > 1$, and thus $(u,v) \in E_1$ with probability $1$.

Alternatively, suppose at least $1/3$ of the flow is sent through paths of length $2$.  Let $w_1, w_2, \dots, w_{\ell}$ be the midpoints of these paths (note that all $w_i$ are distinct).  For each $i \in [\ell]$, let $f_i = \min\{x_{u,w_i}, x_{w_i, v}\}$.  For a length $2$ path to be in $E'$, it is sufficient for one of these $w_i$ to have $T_{w_i} \leq \rho f_i$ (since that would mean both $(u, w_i)$ and $(w_i, v)$ would be in $E_1$).  If $f_i \geq 1/\rho$ for some $i$, then this happens with probability $1$.  If $f_i < 1/\rho$, then the probability that none of these events occur is $\prod_{i=1}^{\ell} (1-\rho f_i) \leq e^{-\rho \sum_i f_i}$.  Since $1/3$ flow was sent along these paths, $\sum_i f_i \geq 1/3$, and thus the probability that $E_1$ does not contain length two $u-v$ path is at most $e^{-\rho / 3}$, which is clearly small enough to satisfy the lemma.

The most difficult case is when at least $1/3$ of the flow is sent through paths of length $3$.  For each such path $P$, let $p_1$ denote the first vertex after $u$ and let $p_2$ denote the second vertex after $u$, so $P = u \rightarrow p_1 \rightarrow p_2 \rightarrow v$.  For every edge $(w,y)$, let $\hat{x}_{w,y}$ be the amount of $u\rightarrow v$ flow using paths of length $3$ that use $(w,y)$.  Clearly $\hat x_{w,y} \leq x_{w,y}$, so it is sufficient to show that the rounding algorithm works when using the $\hat x$ values instead of the $x$ values.  We now divide this case into five subcases.

\paragraph{Case 0:} We first consider the case that there is some edge $(u,a)$ with $\hat x_{u,a} \geq 1/\rho$.  Then $(u,a)$ is in $E_1$ with probability $1$, so it is definitely included in the spanner, and at least $1/\rho$ flow actually flows through $(u,a)$.  So this is essentially like the case of paths of length $2$, just with flow of $1/\rho$ instead of $1/3$.  Let $P_a$ denote the set of length $3$ paths that begin with the edge $(u,a)$ (i.e.~the set of paths where $p_1 = a$).  If there is some path $P \in P_a$ with $\hat x_{a, p_2} \geq 1/\rho$, then clearly by flow conservation $\hat x_{p_2, v} \geq 1/\rho$, so with probability $1$ all of $P$ is in $E_1$.  Otherwise, the probability that such a $P$ is not contained in $E_1$ is at most $1-\rho x_{a,p_2}$.  Thus the probability that we get no $P \in P_a$ is at most $\prod_{P \in P_a} (1-\rho x_{a,p_1}) \leq e^{-\rho \sum_{P \in P_a}x_{a, p_2}} = e^{-\rho \sum_{P \in P_a} f(P)} \leq 1/e$, satisfying the lemma. The same argument can be made for the case that there is some edge $(b,v)$ with $\hat x_{b,v} \geq 1/\rho$, and if some other edge $(w,y)$ has $\hat x_{w,y} \geq 1/\rho$ then by flow conservation there must be some $e = (u,a)$ or $e = (b,v)$ with $\hat x_e \geq 1/\rho$.  So for the rest of the argument we assume without loss of generality that $\hat x_e \leq 1/\rho$ for all $e \in E$.

\paragraph{Cases 1-4:} We now divide the length $3$ paths into four types:
\begin{enumerate}
\item $\sqrt{\rho \hat x_{p_1, p_2}} \leq \rho \hat x_{u, p_1}$ and $\sqrt{\rho \hat x_{p_1, p_2}} \leq \rho \hat x_{p_2, v}$
\item $\rho \hat x_{u, p_1} \leq \sqrt{\rho \hat x_{p_1, p_2}}$ and $\rho \hat x_{p_2, v} \leq \sqrt{\rho \hat x_{p_1, p_2}}$
\item $\rho \hat x_{u, p_1} \leq \sqrt{\rho \hat x_{p_1, p_2}} \leq \rho \hat x_{p_2, v}$
\item $\rho \hat x_{p_2, v} \leq \sqrt{\rho \hat x_{p_1, p_2}} \leq \rho \hat x_{u, p_1}$
\end{enumerate}

These types are exhaustive, so since at least $1/3$ flow uses these paths at least one of the types contains paths that correspond to at least $1/12$ units of flow.  Let $\mathcal P_i$ denote the paths of type $i$.  We now consider each of the four types in turn.

\paragraph{Case 1:} Suppose that at least $1/12$ units of flow use paths from $\mathcal P_1$.  For every edge $(w,y)$ appearing in some path from $\mathcal P_1$, let $x'_{w,y}$ be the actual amount of $\mathcal P_1$ flow that uses $(w,y)$.  Then $x'_{w,y} \leq \hat x_{w,y} \leq x_{w,y}$ for all $(w,y)$.  For any path $P \in \mathcal P_1$, a sufficient condition for $P$ to appear in $E'$ is for $T_{p_1} \leq \sqrt{\rho x'_{p_1, p_2}}$ and for $T_{p_2} \leq \sqrt{\rho x'_{p_1, p_2}}$, since if this happens then $(u,p_1) \in E_1$, $(p_1, p_2) \in E_2$, and $(p_2, v) \in E_1$.  For each $P$, let $E_P$ be the event that this happens.  Then $\Pr[E_P] = \sqrt{\rho x'_{p_1, p_2}} \sqrt{\rho x'_{p_1, p_2}} = \rho x'_{p_1, p_2}$.  Let $Y = \sum_{P \in \mathcal P_1} E_P$ be the number of type $1$ paths that are in $E$ because of this sufficient condition.  It is sufficient to show that the probability that $Y > 0$ is at least a constant, since the repetition of the algorithm $C \log n$ times makes this probability become at least $1 - 1/n^3$.  To bound the probability that $Y > 0$ we will use Janson's inequality~\cite[Chapter 3]{DP09}.  Janson's inequality has two parameters: the expectation $\EX[Y]$ and a value $\Delta$ which intuitively measures the amount of dependency.  Informally, $\Delta$ is the sum over (ordered) pairs of dependent events of the probability that they both occur.  In our setting, two events $E_P$ and $E_{P'}$ are dependent if $P$ and $P'$ share either a first edge or a last edge (if they share a middle edge then they are obviously the same path).  Let $P \sim P'$ if $E_P$ and $E_{P'}$ are dependent.  Then for us $\Delta = \sum_{P \in P_1} \sum_{P' \in P_1 : P' \sim P} \Pr[E_P \land E_{P'}]$.  Janson's inequality implies that $\Pr[Y = 0] \leq e^{-\frac{\EX[Y]^2}{\EX[Y] + \Delta}}$.

Every length $3$ path has a different middle edge, so the fact that at least $1/12$ units of flow use paths in $\mathcal P_1$ implies that $\EX[Y] = \sum_{P \in \mathcal P_1}\rho x'_{p_1, p_2} \geq \rho / 12$.  To bound $\Delta$, we first consider the case the two paths share the first edge, i.e.~we will try to bound $\Delta_1 = \sum_{a \in V} \sum_{P \in P_a} \sum_{P' \neq P \in P_a} \Pr[E_P \land E_{P'}]$, where $P_a$ is the set of paths in $P_1$ that begin with the edge $(u,a)$.  Let $P, P' \in P_a$ be two such paths.  In order for both $E_P$ and $E_{P'}$ to occur, it is necessary and sufficient for the following four conditions to hold: (1) $T_a \leq \sqrt{\rho x'_{a, p_2}}$, (2) $T_a \leq \sqrt{\rho x'_{a, p'_2}}$, (3) $T_{p_2} \leq \sqrt{\rho x'_{a, p_2}}$, and (4) $T_{p'_2} \leq \sqrt{\rho x'_{a, p'_2}}$.  Obviously the probability of this happening is exactly $\sqrt{\rho x'_{a, p_2}} \cdot \sqrt{\rho x'_{a, p'_2}} \cdot \min\{\sqrt{\rho x'_{a, p_2}}, \sqrt{\rho x'_{a, p'_2}}\}$.  Using the geometric mean to upper bound the minimum, we get that $\Pr[E_P \land E_{P'}] \leq \rho^{3/2} (x'_{a, p_2})^{3/4} (x'_{a, p'_2})^{3/4}$.  Thus
\begin{align*}
\Delta_1 & = \sum_{a \in V} \sum_{P \in P_a} \sum_{P' \neq P \in P_a} \Pr[E_P \land E_{P'}] \\
& \leq \rho^{3/2} \sum_{a \in V} \sum_{P \in P_a} \sum_{P' \neq P \in P_a} (x'_{a, p_2})^{3/4} (x'_{a, p'_2})^{3/4} \\
& \leq \rho^{3/2} \sum_{a \in V} \left(\sum_{P \in P_a} \left(x'_{a, p_2}\right)^{3/4}\right)^2 \\
& \leq \rho^{3/2} \sum_{a \in V} \left(n \cdot \left(\frac{x'_{u,a}}{n}\right)^{3/4}\right)^2 = \rho^{3/2} \sqrt{n} \sum_{a \in V} \left(x'_{u,a}\right)^{3/2} \\
& \leq \rho^{3/2} \sqrt{n} \cdot \rho \left(\frac{1}{\rho}\right)^{3/2} = \rho \sqrt{n}
\end{align*}

The same analysis can be done for the other type of dependent paths, when they agree on the last edge.  Thus $\Delta \leq 2 \rho \sqrt{n}$.  So Janson's inequality gives us that $\Pr[Y = 0] \leq e^{-\frac{\rho^2 / 144}{\rho/12 + 2 \rho \sqrt{n}}}$.  By setting $\rho$ to be larger than $288\sqrt{n} + 12$ we get that this probability is less than $1/e$.

\paragraph{Case 2:} We will use a different style of analysis for the case when at least $1/12$ units of flow use paths in $\mathcal P_2$.  Instead of bounding the expectation and using a concentration bound, we will simply show that with constant probability Algorithm~\ref{alg:3spanner} includes a $u\rightarrow v$ path.  As in the previous case, let $x'_{w,y}$ denote the actual amount of flow on $\mathcal P_2$ paths using edge $(w,y)$.  Then $x'_{w,y} \leq \hat{x}_{w,y} \leq x_{w,y}$, so it suffices to prove that rounding using the $x'$ values works with constant probability.  Note that for any path $P \in \mathcal P_2$ the middle edge is used only by $P$, so $x'_{p_1, p_2} = \hat x_{p_1, p_2}$.  We will show that with constant probability there is a path $P \in \mathcal P_2$ with $T'_{p_1} \leq\rho x'_{u, p_1}$ and $T_{p_2} \leq \rho x'_{p_2, v}$.  It is easy to see that this is sufficient for all three edges of $P$ to be in $E'$.

Let $B \subseteq V$ be the set of vertices that are on the last hop of a path in $\mathcal P_2$, i.e.~$b \in B$ if there is some $P \in \mathcal P_2$ with $p_2 = b$.  Let $B' = \{b \in B : x'_{b, v} \geq \frac{1}{24 n}\}$ be the set of vertices from $B$ with large flow through them, and let $\mathcal P'_2 = \{P \in \mathcal P_2 : p_2 \in B'\}$ be the paths that go through these vertices.  Since $1/12$ flow total is sent on $\mathcal P_2$ paths, and each vertex not in $B'$ can transport at most $1/(24n)$ flow, at least $1/24$ flow is sent on $\mathcal P'_2$ paths.  Note that we are still in the case where every edge $e$ has $\hat{x}_e \leq 1/\rho$, so $1/(24 n) \leq x'_{b, v} \leq 1/\rho$ for all $b \in B'$.  For every $b \in B'$ let $E_b$ be the event that $T_b \leq \rho x'_{b,v}$ and let $Y = \sum_{b \in B'} E_b$ be the number of vertices in $B$ for which this event occurs.  Then $\EX[Y] = \sum_{b \in B'} \rho x_{b,v} = \rho \sum_{b \in B'} x_{b,v} \geq \rho / 24$, where the last inequality follows from the fact that at least $1/24$ flow uses paths in $\mathcal P'_2$.  Since these events are all independent, a simple Chernoff bound implies that $Y \geq \rho / 48$ with probability at least $1 - e^{-\rho / 192}$.  Let $\hat B \subseteq B'$ be the set of vertices in $B'$ for which this even occurs, so with high probability $|\hat B| \geq \rho / 48$.

Now we want to lower bound the probability that at least one path that passes through $\hat B$ has a corresponding first hop that is below the threshold.  Let $A$ be the set of first hops for vertices in $\hat B$, so $a \in A$ if and only if there is some path $P \in \mathcal P'_2$ and $b \in \hat B$ with $a = p_1$ and $b = p_2$.  It suffices for at least one $a \in A$ to have $T'_a \leq \rho x'_{u,a}$, since by the definition of $A$ there is some corresponding $b \in \hat{B}$ that completes a path in $\mathcal P'_2$.  Note that since $|\hat B| \geq \rho / 48$ and every $b \in \hat B$ has $x'_{b,v} \geq 1/(24 n)$, the total amount of flow passing through paths in $\mathcal P'_2$ that use vertices in $\hat{B}$ as their final hops is at least $\rho / 48 \cdot (1/(24 n)) = \rho / (1152 n)$.  Thus $\sum_{a \in A} x'_{u,a} \geq \rho / (1152 n)$, so the probability that no $a \in A$ has $T'_a \leq \rho x'_{u,a}$ is at most $\prod_{a \in A} (1- \rho x'_{u,a}) \leq e^{- \rho \sum_{a \in A} x'_{u,a}} \leq e^{- \rho^2 / (1152 n)}$.  Note that here we use the fact that that the $T'$ thresholds are independent of the $T$ thresholds, since otherwise this probability calculation for $A$ could be dependent on the already chose thresholds of $\hat{B}$.

The total probability that we do not include some length $3$ path is at most the sum of the probability that $|\hat B|$ is not large enough and the probability that no corresponding $a$ is selected, which is at most $e^{-\rho / 192} + e^{-\rho^2 / (1152 n)} < 1/e$ as desired.

\paragraph{Case 3:} In the third case there is at least $1/12$ units of flow along paths in $\mathcal P_3$.  We will use an analysis similar to that of the second case.  Let $\hat P_3 \subseteq \mathcal P_3$ be the collection of type $3$ paths with large first hop capacity: $P \in \hat P_3$ if $\hat x_{u,p_1} \geq 1/(24 n)$.  Obviously at most $1/24$ flow can use low capacity first hops, so at least $1/24$ units of flow uses paths in $\hat P_3$.  If $P \in \hat P_3$, then because $P$ is type $3$ we know that $\sqrt{\rho \hat x_{p_1, p_2}} \geq \rho \hat x_{u, p_1}$ and thus $\hat x_{p_1, p_2} \geq \rho \hat x_{u,p_1}^2 \geq \rho /(576 n)$.  So a sufficient condition for a path $P \in \hat P_3$ to be in $E'$ is for $T'_{p_1} \leq \rho \hat x_{u,p_1}$ and for $T_{p_2} \leq \sqrt{\rho^2 / (576 n)^2} = \rho / (24 n)$.

As before, let $B = \{b \in V : \exists P \in \hat P_3 \text{ with } b = p_2\}$.  For $b \in B$, let $x'_{b, v} = \sum_{P \in \hat P_3 : b = p_2} f(P)$ be the amount of flow along paths in $\hat P_3$ that use $b$ as a last hop.  So $x'_{b,v} \leq \hat x_{b,v}$.  Let $B' = \{b \in B : x'_{b,v} \geq 1/(48 n)\}$.  Since at least $1/24$ flow is sent using paths in $\hat P_3$, at least $1/48$ flow is sent using paths in $\hat P_3$ that use a vertex in $B'$ as a last hop.  Call this set of paths $P'_3$.  Now we partition $B'$ into classes based on their $x'$ values: let $B_i = \{ b \in B' : 1/2^{i} < x'_{b,v}  \leq 1/2^{i-1}\}$.  Note that  the first $\log \rho$ of these classes are empty since we are still in a setting where all edges $e$ have $\hat x_e \leq 1/\rho$, and there are at most $\log n$ classes by the definition of $B'$.  These classes also partition the paths in $P'_3$ (since every path in $P'_3$ uses a vertex from $B'$ as a final hop), so at least one class contains at least $1/(48 \log n)$ flow.  Let $B_i$ be this class.  Then every vertex $b \in B_i$ has $x'_{b, v} \geq 1/2^i$ and $|B_i| \geq (1/(48 \log n) / (1/2^{i-1}) = 2^{i-1} / (48\ \log n) \geq \rho / (48 \log n)$.

Let $\hat B \subseteq B_i$ be the set of vertices $b \in B_i$ with $T_b \leq \rho / (24 n)$.  The expected number of vertices in $\hat B$ is $|B_i| \rho / (24 n) \geq \rho^2 / (1152 n \log n)$.  Since $\rho \geq \Omega(\sqrt{n} \log n)$ this becomes $\Omega(\log n)$, so a simple Chernoff bound suffices to guarantee that $|\hat B| \geq (1/2) \EX[|\hat B|] \geq |B_i| \rho / (48 n) \geq 2^{i-1} \rho / (2304 n \log n)$ with probability at least $1 - e^{-\Theta(\rho^2 / n \log n)} \geq 1-e^{-\Theta(\log n)}$.  Assuming that this occurs, since $x'_{b,v} \geq 1/2^i$ for all $b \in B_i$ the amount of flow along paths in $P'_3$ that use a vertex in $\hat B$ as their last hop is at least $\rho / (4608 n \log n)$.  Let $A$ be the set of vertices that are the first hops of these paths, i.e.~$A = \{a \in V: \exists P \in P'_3,\ b \in \hat B \text{ with } a = p_1 \text{ and } b = p_2\}$.  Then $\sum_{a \in A} \hat x_{u,a} \geq \rho / (4608 n \log n)$ in order for there to be enough capacity to shop the flow.  Now in order to complete some path from $P'_3$ we just need for one of these $a \in A$ to have $T'_a \leq \rho \hat x_{u,a}$.  The probability that this does not happen is at most $\prod_{a \in A} (1- \rho \hat x_{u,a}) \leq e^{-\rho \sum_{a \in A} \hat x_{u,a}} \leq e^{-\rho^2 / 4608 n \log n}$.

Thus the total probability of failure is at most $2 e^{-\Theta(\rho^2 / n \log n)}$ (the probability that $|\hat B|$ is below half of its expectation plus the probability that no $A$ vertex has low enough threshold).  Since $\rho \geq \Theta(\sqrt{n} \log n)$ this becomes at most a constant, satisfying the lemma.

\paragraph{Case 4:} In this case at least $1/12$ units of flow use paths in $\mathcal P_4$.  This case is completely analogous to case 3, since we did not ever use directionality in our proof of case 3.  Since this is the final case, it concludes the proof of the lemma.
\end{proof}

It is interesting to note that Lemmas~\ref{lem:3spanner_cost} and~\ref{lem:3spanner_correct} hold even for a weighted version in which every edge has an arbitrary nonnegative cost associated with it and our goal is to find the minimum cost $3$-spanner.  So our approximation algorithm actually generalizes to this weighted version.

\begin{theorem} \label{thm:3spanner_thm}
There is a polynomial time $\tilde{O}(\sqrt{n})$-approximation algorithm for the unit-length directed $3$-spanner problem, even with arbitrary costs on the edges.
\end{theorem}
\begin{proof}
The algorithm is simple: solve LP~\eqref{LP:spanner1} (note that this can be done exactly since in this setting the linear program has only a polynomial number of variable) and then repeat Algorithm~\ref{alg:3spanner} with fresh randomness $O(\log n)$ times.  Lemma~\ref{lem:3spanner_correct} implies that this gives a valid solution with high probability, and Lemma~\ref{lem:3spanner_cost} implies that it is a $\tilde{O}(\sqrt{n})$-approximation.
\end{proof}

\subsection{Extension to Fault-Tolerant Version} \label{sec:stretch3_FT}
It is easy to see that this algorithm can be trivially extended to the $r$-fault-tolerant setting.  For each set $F$ of faults, the analysis works the same as in Lemma~\ref{lem:3spanner_correct}.  We just need to solve LP~\eqref{LP:FT_spanner} instead of LP~\eqref{LP:spanner1} and change the parameters ($\rho$ and the number of times Algorithm~\ref{alg:3spanner} is repeated) to make the probability of failure small enough to apply a union bound to all possible failure sets $F$ and edges $(u,v)$ instead of just over the edges.  In particular, for the vertex failure setting we need the probability of failure to be less than $1/n^{r+2}$ and for the edge failure setting we need to probability of failure to be less than $1/n^{2r+2}$.  The main takeaway is that for directed $3$-spanner we get fault-tolerance at a cost of only $\tilde{O}(r)$ instead of something exponential in $r$,  as for the $k > 3$ case and previous work on absolute bounds~\cite{CLPR09}.

\begin{theorem} \label{thm:3spanner_FT}
For any constant $\epsilon > 0$ there is a polynomial time $(\frac{1}{3(1+\epsilon)}, O(\frac{1+\epsilon}{\epsilon} r \sqrt{n} \log^2 n))$-approximation algorithm for the $r$-fault-tolerant directed $3$-spanner problem with unit lengths.  There is also a true $O(r \sqrt{n} \log^2 n)$-approximation algorithm that takes $n^{O(r)}$ time.
\end{theorem}
\begin{proof}
Let $\Delta$ be the number of times that Algorithm~\ref{alg:3spanner} is repeated.  The analysis of Lemma~\ref{lem:3spanner_correct} has the following failure probabilities for the different cases.  If at least $1/3$ flow is on the direct edge, then the probability of failure is $0$.  If at least $1/3$ flow is on $2$-hop paths, then the probability of failure is at most $e^{-\Delta\rho / 3}$.  If at least $1/3$ flow is on the $3$-hop paths, then we have five subcases.  For the zeroth case, if there is some edge with $\hat x_e \geq 1/\rho$ then the probability of failure is at most $(1/e)^{\Delta}$.  For case $1$, the failure probability (i.e.~the probability that $Y= 0$, which we bound by Janson's inequality) is at most $e^{-\Delta \cdot \Theta(\rho / \sqrt{n})}$.  For case $2$ the failure probability is at most $(e^{-\Theta(\rho)} + e^{-\Theta(\rho^2 / n)})^{\Delta}$, where the first term is from the Chernoff bound guaranteeing that $|\hat B|$ is large and the second term is from the probability that none of the $A$ vertices have low enough thresholds.  For the third case, the probability of failure is at most $(e^{-\Theta(\rho^2 / (n \log n))} + e^{-\Theta(\rho^2 / (n \log n))})^{\Delta}$, where the first term is from the Chernoff bound to guarantee that $|\hat B|$ is at least half its expectation and the second is from  the probability that none of the $A$ vertices have low enough thresholds.  Finally, the fourth case is symmetric to the third.

It is easy to verify that we set $\Delta = C_1 (2r+2) \log n$ and $\rho = C_2 \sqrt{n} \log n$ then all of these failure probabilities are less than $1/n^{2r+2}$, which is what we needed for the edge fault setting (we can set $\Delta = C_1 (r+2) \log n$ for the vertex fault setting).  The total approximation that we get is thus $\Delta \rho = O(r \sqrt{n} \log^2 n)$.  Combining this with Theorem~\ref{thm:FT_LP_solve} completes the proof.
\end{proof}

\section{LP-based approximation of $2$-spanner}
\label{sec:2-spanner}

The $2$-spanner problem is qualitatively and quantitatively different from $k$-spanner with $k > 2$: it is known~\cite{KP94,EP01} that it can be approximated to $O(\log n)$ and that this is tight~\cite{Kortsarz01} (assuming $P \neq NP$).
Note that this approximation algorithm is only known to work
for the unit-length version.
We first show that our LP relaxation \eqref{LP:spanner1} has integrality gap
of $\Theta(\log n)$, and thus offers comparable approximation ratio.
We further show that our approach,
namely the flow-based LP relaxation and a direct rounding procedure,
easily adapts to bounded-degree case
which was studied in the literature \cite{KP94,DK99,EP01},
and also to the fault-tolerant case,
for which no approximation was previously known.

\subsection{Lower Bound on the Integrality Gap}
\label{sec:2-spannerLB}
We first show that the integrality gap is $\Omega(\log n)$.  The intuition is that we will apply the hardness reduction from set cover to $2$-spanner to an instance of set cover that has a large integrality gap.  We actually prove the gap for the more general setting of undirected graphs; it is easy to see that this implies the same gap for directed graphs. We first describe the generic reduction, then the particular set cover instance that we apply it to.

\begin{theorem} \label{thm:2-spannerLB}
The integrality gap of LP~\eqref{LP:spanner1} for undirected unit-length $2$-spanner is $\Omega(\log n)$
\end{theorem}
\begin{proof}
Suppose we have a (unweighted) set cover instance with elements $U$ and sets $\mathcal{S}$, where $|U| = N$ and $|\mathcal{S}| = M$.  We create a graph $G$ with vertex set $U \cup \mathcal{S} \cup \{x_i: i \in [k]\}$, where $k =M^2$.  In other words, there is a vertex for every element, a vertex for every set, and $k$ new vertices $x_1, \dots, x_k$.  Clearly the number of vertices is polynomial in the size of the set cover instance (it is in fact $n = M^2 + M + N$).  There is an edge from every $x_i$ to every set node and to every element node, an edge between every two set nodes, and an edge between a set node $S \in \mathcal{S}$ and every $e \in U : e \in S$.  More formally, the edge set is $\{\{x_i, S\} : i \in [k], S \in \mathcal{S}\} \cup \{\{x_i, e\} : i \in [k], e \in U\} \cup \{\{S, S'\} : S, S' \in \mathcal{S}\} \cup \{\{S, e\} : s \in \mathcal{S}, e \in U, e \in S\}$.

The set cover instance that we use has element set $\mathbb{F}_2^q \setminus \{\vec{0}\}$, so there are $2^q - 1$ elements.  There is a set $S_{\alpha}$ for every $\alpha \in \mathbb{F}_2^q$ (so there are $2^q$ sets), where $S_{\alpha} = \{e \in \mathbb{F}_2^q \setminus \{\vec{0}\} : \alpha \cdot e = 1\}$.  We are using the normal notion of dot product over $\mathbb{F}_2^q$, i.e.~$\alpha \cdot e = \alpha_1 e_1 + \dots + \alpha_q e_q$ (mod $2$).  It is easy to see that every element is in exactly half of the sets.  This large amount of overlap intuitively allows the linear program to ``cheat''.

To see that the LP has a small solution, we will set the capacity of the edges between set vertices to $1$, the edges between set vertices and element vertices to $1$, and the edges between $x_i$ vertices and element vertices to $0$.  We will also set the capacity of edges between $x_i$ vertices and set vertices to $2/M$.  Obviously this solution has cost at most $kM\frac{2}{M} + M^2 + MN = O(M^2) = O(n)$, so it remains to show that it is a feasible solution.  To show this, for every edge in the original graph we need to find a way to route at least one unit of flow subject to our capacities from one endpoint to the other along paths of length at most $2$.  This is trivial for every edge that we set to have capacity $1$, so we just need to worry about edges incident on $x_i$ nodes.  For edges of the form $\{x_i, S\}$ with $S \in \mathcal{S}$, we can send $1/M$ flow on every edge from $x_i$ to $\mathcal{S}$ (including the edge from $x_i$ to $S$, and then flow that was set to sets $S' \neq S$ can be forwarded along the $\{S',S\}$ edge.  For edges of the form $\{x_i, e\}$ with $e \in U$, we can send $2/M$ flow from $x_i$ to every set that contains $e$.  Since exactly half of the sets contain $e$ this adds up to a total flow of $1$.  This flow can then be forwarded directly to $e$, since there is an edge of capacity $1$ between $e$ and every set containing $e$.  Thus this if a feasible solution to the flow LP of cost $O(n)$.

Now we want to show than any integral solution has cost at least $\Omega(n \log n)$.  Consider some arbitrary integral solution (i.e. a setting of 0/1 capacities to every edge such that one unit of flow can be sent between the endpoints of any original edge using paths of length at most $2$).  Consider an edge $\{x_i,e\}$ with $e \in U$.  Either this edge has capacity $1$, or there is some $S \in \mathcal{S}$ with $e \in S$ such that the edges $\{x_i, S\}$ and $\{S, e\}$ both have capacity $1$.  This is because the only paths of length at most $2$ between $x_i$ and $e$ are paths of this form and the one direct edge.  Since this is true for every $e$, the vertices adjacent to $x_i$ must form a valid set cover of this original instance (where an edge directly to an element $e$ is equivalent to adding the set $\{e\}$).  Thus the degree of any $x_i$ node must be at least the size of the smallest valid set cover.  For our set cover instance, it is easy to see that the size of the smallest cover is at least $q$.  To see this, suppose otherwise, i.e.~assume there is some collection of sets $S_{\alpha_1}, \dots, S_{\alpha_{q-1}}$ that covers the elements.  Then $\cap_{i=1}^{q-1} \overline{S_{\alpha_i}} =\emptyset$, so $\cap_{i=1}^{q-1} \{e \in \mathbb{F}_2^q: \alpha_i \cdot e = 0\} = \{\vec{0}\}$.  But this is a contradiction, since the intersection of $q-1$ hyperplanes in the $q$-dimensional vector space over $\mathbb{F}_2$ cannot be just a single point (that would require at least $q$ hyperplanes).  So any valid set cover has size at least $q \geq \log N$.

So now we know that any integral solution to the flow LP has cost at least $kq \geq M^2 \log N = \Omega(n \log n)$, thus proving that the integrality gap of the flow LP is at least $\Omega(\log n)$.
\end{proof}

\subsection{Upper Bounds  via Direct Rounding}
\label{sec:2-spannerUBdirect}
We bound the integrality gap of the LP by randomized rounding of
the LP solution.
We first show a bound of $O(\log n)$,
and then refine it to $O(\log \Delta)$,
where $\Delta$ is the maximum degree of the graph.
Both results hold also for directed graphs, and easily extend to the
client-server version and to the augmentation version of \cite{EP01}
(where part of the spanner is already given).
Finally, we explain how these two bounds easily extend to
the $r$ fault-tolerant version,
losing only a factor of $r$ in the approximation guarantee.

\subsubsection{General digraphs (with unit length edges)}

\begin{theorem}
  \label{thm:UBdirect1}
  For the directed $2$-spanner problem with unit lengths, (even the version with edge
  costs), the LP relaxation \eqref{LP:spanner1} has integrality gap $O(\log n)$.
\end{theorem}

\begin{proof}
Consider a directed graph $G=(V,E)$ with edge costs $c_e\ge 0$
(but unit edge lengths) and a feasible solution to LP \eqref{LP:spanner1}.
The argument for undirected graphs is identical.
We employ the rounding procedure depicted in Algorithm \ref{alg:2spanner},
where $C>0$ is a sufficiently large constant to be determined later.  This is basically a simplified version of the rounding algorithm that we used for directed unit-length $3$-spanner.

\begin{algorithm}[H]
\caption{Rounding algorithm for $2$-spanner.}
\label{alg:2spanner}
Set $\rho=C\ln n$.
\\
For every $v\in V$ choose independently a random threshold $T_v\in[0,1]$.
\\
Output $E'=\{(u,v)\in E:\ \minn{T_u,T_v} \leq \rho\cdot x_{u,v}\}$.
\end{algorithm}

The output of this algorithm has expected total cost
\[
  \EX[\sum_{e\in E'} c_e]
  = \sum_{(u,v)\in E} c_{u,v} \Pr[\minn{T_u,T_v}
  \leq \rho\cdot x_{u,v}]
  \le 2\rho\sum_{e\in E} c_ex_e.
\]
By Markov's inequality, with probability at least $2/3$,
the cost of the output solution exceeds the LP by factor
$\leq 3\rho \leq O(\log n)$.

We proceed to show that the output forms a $2$-spanner (with high probability);
recall that it suffices to verify the stretch for  \emph{edges} $(u,v)\in E$.
So fix $(u,v)\in E$ and consider a $u-v$ path of length $2$,
denoted henceforth $P_z=(u,z,v)$.
A similar notation but with $z=\bot$ refers to a path of length $1$,
i.e. $P_\bot=(u,v)$.
By the LP constraints, the sum of flows along all these paths is
\[
  \sum_{z\in V\cup\{\bot\}} f_{P_z} \geq 1.
\]
Here and in the sequel, notation like $z\in V$ implicitly excludes
any $z\in V$ for which there is no such path $P_z=(u,z,v)$.
If $f_{P_\bot}\ge 1/2$,
i.e. at least half of this flow is routed directly along the edge $(u,v)$,
then by the LP constraints $x_{u,v}\ge 1/2$
and with probability $1$ the edge $(u,v)$ is included in $E'$.
Otherwise, $\sum_{z\in V} f_{P_z}\ge 1/2$;
now observe that whenever $T_z/\rho \leq f_{P_z}$,
the algorithm's output $E'$ contains the entire $2$-path $P_z$
(because by the LP constraints $\minn{x_{uz},x_{zv}} \ge f_{P_z} \ge T_z/\rho$).
Put in the contrapositive,
for $E'$ to contain no such $2$-path $P_z$,
the event $\{\forall z\in V,\ T_z/\rho > f_{P_z} \}$ must occur,
which happens with probability,
\begin{equation}
  \label{eq:StretchProb1}
  \Pr[\forall z\in V,\ T_z/\rho > f_{P_z} ]
  = \prod_{z\in V} (1-\rho\cdot f_{P_z})
  \leq e^{-\rho\sum_{z\in V} f_{P_z}}
  \leq n^{-C/2}.
\end{equation}
Now setting $C=6$ and taking a union bound over less than $n^2$
edges $(u,v)\in E$,
we get that with probability at least $1-1/n$,
the algorithm outputs a $2$-spanner of $G$.

Taking a union bound, we conclude that with probability at least $1/2$,
the algorithm outputs a $2$-spanner of cost at most the LP value times
$O(C\log n)$, proving Theorem \ref{thm:UBdirect1}.
\end{proof}

\subsubsection{Refinement to digraphs of bounded-degree}

\begin{theorem}
  \label{thm:UBdirect2}
  For the directed $2$-spanner problem with unit lengths on graphs of maximum (in
  and out) degree at most $\Delta\ge 2$, the LP relaxation
  \eqref{LP:spanner1} has integrality gap $O(\log \Delta)$.
\end{theorem}
%\rnote{Can we improve the bound to $O(\log(|E|/|V|)$ i.e. average degree,
%and match \cite{KP94}?}
%\rnote{Unlike the $O(\log n)$ bound, this proof only works for unit costs,
%essentially because we analyze the cost using local lemma and Chernoff
%(instead of one global Markov's inequality)}
\begin{proof}
Suppose now that (in and out) vertex-degrees in $G$ are at most $\Delta\ge 2$.
Consider the same rounding algorithm, except that $\rho=C\log\Delta$,
and let us show that it succeeds with positive (but possibly small) probability.
We shall require the following symmetric form of the Lov\'asz Local Lemma
(see, e.g., \cite{AS00}).
\begin{lemma}[\Lovasz Local Lemma]
\label{lem:lll}
Let $A_1, \ldots, A_n$ be events in an arbitrary probability space.
Suppose that each $A_i$ is mutually independent of all but at most $d$
other events $A_j$,
and suppose that $\Pr[A_i] \leq p$ for all $1 \leq i
\leq n$.  If $ep(d+1) \leq 1$ then $\Pr[\wedge_{i=1}^n \overline{A_i}] > 0$.
\end{lemma}
\REMOVE{
We shall require the following general form of the Lov\'asz Local Lemma
(see, e.g., \cite{AS00}).
\begin{lemma}[\Lovasz Local Lemma]
\label{lem:lll}
Let $A_1, \ldots, A_N$ be events in an arbitrary probability space,
and let $D$ be a dependency digraph for them,
i.e. each $A_i$ is mutually independent of all the events
$\{A_j:\ (i,j)\notin E(D)\}$.
Suppose there are $p_1,\ldots,p_n\in [0,1)$ such that
$\Pr[A_i] \leq p_i \prod_{(i,j)\in E(D)} (1-p_j)$ for all $i$.
Then
$$\textstyle
  \Pr[\wedge_{i=1}^N \overline{A_i}] \ge \prod_{i=1}^N (1-p_i) > 0.$$
\end{lemma}
}

For an edge $(u,v)\in E$, let $A_{u,v}$ be the event that
the output $E'$ contains no $u-v$ path of length at most $2$.
For a given edge $(u,v)\in E$, the same analysis as above shows that
\begin{equation}
  \label{eq:StretchProb2}
  \Pr[A_{u,v}]
  \leq e^{-\rho \sum_{z\in V} f_{P_z}}
  \leq \Delta^{-C/2}.
\end{equation}
Observe that the event $A_{u,v}$ depends only on the random variables
$T_z$ for $z\in N_+(u)$,
where $N_+(u)$ denotes the set of out-neighbors of $u$ in $G$
and $u$ itself (and in our case, $N_+(u)$ includes $v$).
Thus, $A_{u,v}$ is mutually independent of all events $A_{u',v'}$
for which $N_+(u) \cap N_+(u')=\emptyset$,
which by a simple calculation means it is independent of all but at most
$(\Delta+1)^2$ other events $A_{u',v'}$.

The local lemma applies to these events $A_{u,v}$,
but it guarantees a very small positive probability,
which is not enough to bound the cost of the solution $E'$ via a union bound.
Instead, we incorporate the analysis of the cost $|E'|$ into the local lemma,
by ``splitting'' it into multiple local events.
For a vertex $u\in V$, define the random variable $Z_u^+$
to be the number of outgoing edges $(u,v)\in E$
for which $T_v \le \rho\cdot x_{u,v}$.
Define $Z_u^-$ similarly for incoming edges, namely
$\#\aset{(v,u)\in E:\ T_v \le \rho\cdot x_{u,v}}$.
Observe that we can bound the cost of the solution by
\[ \textstyle
  |E'| \leq \sum_{u\in V} (Z_u^+ + Z_u^-),
\]
by simply ``charging'' every edge chosen to $E'$ to one of its endpoints.
Define the event
\[ \textstyle
  B_u
  = \Big\{Z_u^++Z_u^- \geq 4\rho\cdot
           \big(\sum_{(u,v)\in E} x_{u,v} + \sum_{(v,u)\in E} x_{v,u}\big) \Big\}.
\]
We may assume $u$ is incident to at least one edge, say an outgoing edge,
and thus $\sum_{(u,v)\in E} x_{u,v}\ge 1$.
Since $Z_u^+$ is the sum of independent indicators
with expectation $\EX[Z_u^+] \leq \rho \sum_{(u,v)\in E} x_{u,v}$,
we have by a Chernoff bound (see e.g. \cite{MR})
\[ \textstyle
  \Pr\big[Z_u^+ \geq 2\rho \sum_{(u,v)\in E} x_{u,v}\big]
  \leq e^{-(1/4)\rho\sum_{(u,v)\in E} x_{u,v}}
  \leq \Delta^{-C/4}.
\]
Similarly, $\EX[Z_u^-] \leq \rho\sum_{(v,u)\in E} x_{v,u}$
and thus
$
  \Pr\big[Z_u^- \geq 2\rho \maxx{1,\sum_{(v,u)\in E} x_{v,u}}\big]
%  \leq e^{-(1/4)\rho \max{1,\sum_{(v,u)\in E} x_{v,u}}}
  \leq \Delta^{-C/4}
$.
For $B_u$ to occur, at least one of the last two events must occur, hence
\[
  \Pr[B_u] \leq 2 \Delta^{-C/4}.
\]
Observe further that the event $B_u$ depends only on the random variables $T_z$
for $z\in N_+(u)\cup N_-(u)$.

We can now apply the local lemma to all the events $A_{u,v}$ and $B_u$.
Indeed, each of these events is mutually independent of all but
at most $d=O(\Delta^2)$ other events,
and setting $C$ to be a sufficiently large constant and $p=2\Delta^{-C/4}$
yields $p(d+1) \le 1/e$.
Thus, with positive probability none of the events $A_{u,v}$ and $B_u$
occurs, meaning that $E'$ is a $2$-spanner of $G$ and
\[
  |E'|
  \leq \sum_{u\in V} (Z_u^+ + Z_u^-)
  \leq \sum_{u\in V} 4\rho\cdot
           \Big (\sum_{v:(u,v)\in E} x_{u,v} + \sum_{v:(v,u)\in E} x_{v,u} \Big)
  = 8 \mathrm{LP}.
\]
 \end{proof}

\subsubsection{Extension to the fault-tolerant version}

The proofs above easily extends to the fault-tolerant version of the problem,
with respect to the LP relaxation \eqref{LP:FT_spanner}.
For concreteness, we discuss vertex-faults,
but the same arguments hold for edge-faults as well.
Recall that $r$ denotes the maximum number of faults.

The rounding procedure is the same (Algorithm \ref{alg:2spanner}),
except that the factor $\rho$ is increased by a factor of $r$.
For every possible fault set $F\subset V$, $|F|\leq r$,
and every $(u,v)\in E\setminus F$,
we get similarly to \eqref{eq:StretchProb1} that the probability
$E'\setminus F$ does not contain a $u-v$ path of stretch $2$
is $\leq e^{-\rho\sum_{z\in V} f_{P_z}} \leq n^{-Cr/2}$.
Since there are at most $n^{r+2}$ such choices,
we can apply a union bound and prove the following theorem.

\begin{theorem}
  \label{thm:UBdirect1FT}
  For the $r$ fault-tolerant version of the $2$-spanner problem,
  even in the versions with edge costs and directed graphs,
  the LP relaxation \eqref{LP:FT_spanner} has integrality gap $O(r\log n)$.
\end{theorem}

The proof for bounded-degree graphs is similar:
By the same analysis as \eqref{eq:StretchProb2},
for a given $F$ and $(u,v)\in E\setminus F$, we get
$  \Pr[A_{u,v}^F]
  \leq e^{-\rho \sum_{z\in V} f_{P_z}}
  \leq \Delta^{-Cr/2}
$,
and $\Pr[B_u^F] \le 2\Delta^{-Cr/4}$.
Notice that we may restrict this collection of events to
cases where all vertices of $F$ are within distance at most $2$ from $u$,
and then each event is mutually independent of all
but at most $\Delta^{O(r)}$ events.
Applying now the local lemma yields the following theorem.

\begin{theorem}
  \label{thm:UBdirect2FT}
  For the $r$-fault-tolerant version of the $2$-spanner problem
  on directed graphs of maximum (in and out) degree at most $\Delta\ge 2$,
  the LP relaxation \eqref{LP:FT_spanner} has integrality gap $O(r\log \Delta)$.
\end{theorem}

\newpage

\bibliographystyle{alphainit}
\bibliography{robi,spanner}

\end{document}